\documentclass[12pt]{article}
\usepackage[T1]{fontenc} 
\usepackage[utf8]{inputenc}
\usepackage[english]{babel}
\usepackage{enumerate}
\usepackage{verbatim}
\usepackage{graphicx,psfrag,epsf}
\usepackage{amsmath}
\usepackage{amssymb}
\usepackage{amsthm}
\usepackage{mathtools}
\usepackage{color}
\usepackage{multirow}
\usepackage{natbib}
\usepackage{url}
\usepackage{chngcntr}
\usepackage{sectsty}
\sectionfont{\large\centering\uppercase}
\subsectionfont{\noindent\textsc}
\usepackage{hyperref}
\hypersetup{
     colorlinks = true,
     citecolor = blue,
     linkcolor = blue,
     urlcolor = blue
}

\newcommand{\blind}{0}

\newtheorem{theorem}{Theorem}[section]
\newtheorem{lemma}[theorem]{Lemma}

\theoremstyle{definition}
\newtheorem{definition}{Definition}[section]

\theoremstyle{remark}

\numberwithin{equation}{section}

\def\veps{\varepsilon}

\date{}
\begin{document}

\def\spacingset#1{\renewcommand{\baselinestretch}%
{#1}\small\normalsize} \spacingset{1}


\if0\blind
{
  \title{Semismooth Newton Coordinate Descent Algorithm for Elastic-Net Penalized Huber Loss
Regression and Quantile Regression}
  \author{Congrui Yi and Jian Huang \\
  \hspace{.2cm}\\
  Department of Statistics and Actuarial Science \\ University of Iowa}
  \maketitle
} \fi

\if1\blind
{
  \bigskip
  \bigskip
  \bigskip
  \begin{center}
    {\LARGE\bf Semismooth Newton Coordinate Descent Algorithm for Elastic-Net Penalized Huber Loss Regression and Quantile Regression}
\end{center}
  \medskip
} \fi

\bigskip
\begin{abstract}
We propose an algorithm, semismooth Newton coordinate descent (SNCD), for the elastic-net penalized Huber loss regression and quantile regression in high dimensional settings. Unlike existing coordinate descent type algorithms, the SNCD updates each regression coefficient and its corresponding subgradient simultaneously in each iteration. It combines the strengths of the coordinate descent and the semismooth Newton algorithm, and effectively solves the computational challenges posed by dimensionality and nonsmoothness. We establish the convergence properties of the algorithm. In addition, we present an adaptive version of the  ``strong rule" for screening predictors to gain extra efficiency. 
Through numerical experiments, we demonstrate that the proposed algorithm is very efficient and scalable to ultra-high dimensions. We illustrate the application via a real data example.
\end{abstract}

\noindent%

\spacingset{1.45} 

\section{Introduction}




Consider the linear regression model
\[
y_i = \beta_0 + x_i^{\top}\beta + \veps_i
\]
where $x_i$ is a $p$-dimensional vector of covariates,
$(\beta_0, \beta)$ are regression coefficients, and $\veps_i$ is the random error. We are interested in the high dimensional case where $p \gg n$ and the model is sparse in the sense that only a small proportion of the coefficients are nonzero. In such a scenario, a key task is identifying and estimating the nonzero coefficients. A popular approach is the penalized regression

\begin{equation}\label{gen_loss}
\min_{\beta_{0},\beta}\frac{1}{n}\sum_{i}\ell(y_{i}-\beta_{0}-x_{i}^{\top}\beta)+ \lambda P(\beta),
\end{equation}
where $\ell$ is a generic loss function and $p$ is a penalty function with a tuning parameter $\lambda \ge 0$.
We consider the elastic-net penalty \citep{zou2005regularization}
\[
P(\beta) \equiv P_{\alpha}(\beta) = \alpha\|\beta\|_{1}+(1-\alpha)\frac{1}{2}\|\beta\|_{2}^{2}, 0 \le \alpha \le 1,
\]
which is a convex combination of the lasso \citep{tibshirani1996regression} ($\alpha = 1$) and the ridge penalty ($\alpha = 0$).

A common choice for $\ell$ is the squared loss $\ell(t) = t^2/2$, corresponding to the least squares regression in classical regression literature. Although the squared loss is analytically simple, it is not suitable for data in the presence of outliers or heterogeneity. Instead, we could consider two widely used robust alternatives, the Huber loss \citep{huber1973robust} and the quantile loss  \citep{koenker1978regression}.

The Huber loss is
\begin{equation}\label{huberfunc}
\ell(t) \equiv h_{\gamma}(t)=\begin{cases}
\frac{t^{2}}{2\gamma}, & \text{if } |t|\leq\gamma, \\
|t|-\frac{\gamma}{2}, & \text{if } |t|>\gamma,
\end{cases}
\end{equation}
where $\gamma >0$ is a given constant.
This function is quadratic for $|t|\leq \gamma$ and linear for $|t|>\gamma$. In addition, it is convex and first-order differentiable. These features allow it to combine analytical tractability of the squared loss for the least squares and outlier-robustness of the absolute loss for the LAD regression.

The quantile loss is
\begin{equation}\label{quantfunc}
\ell(t) \equiv \rho_{\tau}(t)=t(\tau-I(t<0)), t \in \mathbb{R},
\end{equation}
where $0 < \tau < 1$.
This is a generalization of the absolute loss with $\tau=1/2$.
Rather than the conditional mean of the response given the covariates, quantile regression models conditional quantiles. For heterogeneous data, the functional relationship between the response and the covariates may vary in different segments of its conditional distribution. By choosing different $\tau$, quantile regression provides a powerful technique for exploring data heterogeneity in addition to outlier-robustness.



The theoretical properties of these two regression models have been systematically studied, yet there are relatively few researches on the algorithmic aspect, especially the penalized versions for high-dimensional data.
\cite{holland1977robust} proposed an iteratively re-weighted least squares algorithm for the unpenalized Huber loss regression. However, this algorithm does not have a natural extension to the penalized version. For unpenalized quantile regression, \cite{portnoy1997gaussian} formulated its dual form  as a linear programming problem and proposed an interior point method to solve it. The lasso penalized version can be shown to have a similar dual form, except that it becomes $(n+p)$-dimensional with $p$ extra constraints due to the penalty. Thus it can be solved using the same algorithm and this extension was implemented in
the R package \texttt{quantreg} (\url{http://cloud.r-project.org/package=quantreg}). However, it is not clear if this
approach is scalable to high-dimensional problems. \cite{osborne2011homotopy} proposed a homotopy algorithm for computing solution paths of lasso penalized quantile regression, where the lasso penalty was formulated as a constraint $\sum_{j=1}^p |\beta_j| \leq \kappa$, which is not directly comparable with the unconstrained formulation considered here.


In recent years coordinate descent algorithms have proven to be very effective for pathwise optimization of penalized regression models, see for example,
\cite{friedman2007pathwise} for lasso and fused lasso penalized least squares,  \cite{friedman2010regularization} for elastic-net penalized GLM, and
\cite{breheny2011coordinate} for nonconvex penalized least squares and logistic regression.
The loss functions considered by these authors are either quadratic, or twice differentiable which can be approximated quadratically via Taylor expansion. Hence the coordinate descent iterations have close-form solutions. However, the Huber loss is only first-order differentiable and the quantile loss is nondifferentiable, hence the above approach does not work. \cite{wu2008coordinate} proposed a coordinate descent algorithm for lasso penalized LAD regression that amounts to computing a weighted median at each iteration, but did not provide any guarantee for convergence. Recently \cite{peng2015iterative} proposed a QICD algorithm for nonconvex penalized quantile regression that majorizes the penalty functions by weighted lasso and then solves the problem with coordinate descent. The authors proved convergence of QICD to a stationary point, for which the majorization step plays a critical role. But when the lasso penalty is used, which does not need to be majorized, the algorithm becomes the same as the one in \cite{wu2008coordinate}. In addition, it appears that neither algorithm can be easily generalized to the elastic-net penalty with $0< \alpha < 1$.


In this paper, we propose a novel semismooth Newton coordinate descent (SNCD) algorithm for computing solution paths of the elastic-net penalized Huber loss regression and quantile regression. This algorithm combines the coordinate descent algorithm with the semismooth Newton algorithm (SNA) for solving nonsmooth equations. It is highly efficient and scalable in high-dimensional settings. Unlike a typical coordinate descent method which only updates the primal variable $\beta$, the SNCD utilizes both the primal and the dual information (via subgradient) in its iterations. In addition, an adaptive version of the  strong rule \citep{tibshirani2012strong} for screening predictors is incorporated to gain extra efficiency. We also provide an implementation of SNCD through a publicly available R package \texttt{hqreg} (\url{https://cran.r-project.org/web/packages/hqreg/index.html}) which currently supports the Huber loss, the quantile loss and the squared loss.
This algorithm can be generalized to other problems with nonsmooth loss functions, like the linear support vector machine with the hinge loss.



The rest of this paper is organized as follows.
In section \ref{section_huber} we introduce SNCD for the penalized Huber loss regression and establish its convergence.
 In section \ref{section_quantile} we extend SNCD to penalized quantile regression. Section \ref{section_asr} describes the adaptive strong rule. In section \ref{section_numerical} we investigate the performance of \texttt{hqreg}, our implementation of SNCD, through simulation studies and real datasets.


\section{SNCD for Penalized Huber Loss Regression}
\label{section_huber}


\subsection{Background on Newton Derivatives and SNA}


Based on the concepts of generalized Jacobian \citep{clarke1983optimization} and semismoothness \citep{mifflin1977semismooth}, \cite{qi1993nonsmooth} established superlinear convergence of a Newton-type method for solving finite-dimensional nonsmooth equations, hence the name Semismooth Newton Algorithm (SNA).
The Newton differentiability was introduced later for more general problems including infinite-dimensional cases \citep{chen2000smoothing, ito2008lagrange}. It has a simpler
formulation and is actually a milder condition than semismoothness. Newton derivatives can be calculated
via basic algebra and chain rules as indicated in Lemmas \ref{chain}, \ref{basic} and \ref{piecewise} in Appendix A.

\begin{definition}\label{slant}
A function $F:\mathbb{R}^{m}\rightarrow\mathbb{R}^{l}$ is said to be \emph{Newton differentiable} at $z\in\mathbb{R}^{m}$ if there exists an open neighborhood $\mathcal{N}(z)$ and a mapping $H:\mathcal{N}(z)\rightarrow\mathbb{R}^{l\times m}$ such that $\{H(z+h):z+h\in \mathcal{N}(z), h\neq 0\}$ is uniformly bounded in spectral norm induced by the Euclidean norm and
\[
\|F(z+h)-F(z)-H(z+h)h\|_{2}=o(\|h\|_{2})\quad\text{as}\; h\rightarrow 0.
\]
Here $H$ is called a \emph{Newton derivative} for $F$ at $z$. The set of all Newton derivatives at $z$ is denoted as $\nabla_N F(z)$.
\end{definition}



A function  $F:\mathbb{R}^{m}\rightarrow\mathbb{R}^{l}$ is said to be \emph{locally Lipschitz continuous} at $z$ if there exists $L(z)>0$ such that for all sufficiently small $h$,
\[
\|F(z+h)-F(z)\|_{2}\leq L\|h\|_{2}.
\]
Then $F$
is Newton differentiable at $z$ if and only if $F$ is locally Lipschitz continuous at $z$ \citep{chen2000smoothing}.
This gives a simple characterization of the Newton differentiability.
The following result due to \cite{chen2000smoothing} establishes the superlinear convergence of SNA under the Newton differentiability.


\begin{theorem}\label{superlinear}
Suppose that $F:\mathbb{R}^{m}\rightarrow\mathbb{R}^{m}$ is Newton differentiable at a solution $z^*$ of $F(z)=\mathbf{0}$. Let $H$ be a Newton derivative for $F$ at $z^*$.
Suppose there exists a neighborhood $\mathcal{N}(z^*)$ and $M>0$ such that $H(z)$ is nonsingular and $\|H(z)^{-1}\|\leq M$ for all $z \in \mathcal{N}(z^*)$, then the Newton-type iteration
\[
z^{k+1}=z^{k}-H(z^{k})^{-1}F(z^{k}),\ k=0, 1, \ldots
\]
converges superlinearly to $z^*$ provided that $\|z^0 - z^*\|_2$ is sufficiently small,
where $z^0$ is the initial value.
\end{theorem}

\subsection{Algorithm}
\subsubsection{Description}
Consider the Huber loss $\ell = h_\gamma$, then \eqref{gen_loss} becomes
\begin{equation}\label{huber}
\min_{\beta_{0},\beta} f_H(\beta_0, \beta) = \frac{1}{n}\sum_{i}h_\gamma(y_{i}-\beta_{0}-x_{i}^{\top}\beta)+ \lambda P_{\alpha}(\beta).
\end{equation}

Fix $\lambda$ and $\alpha$, and denote the optimizer by $(\widehat{\beta}_0, \widehat{\beta})$. Since the objective function in \eqref{huber} is convex, $(\widehat{\beta}_0, \widehat{\beta})$
satisfies the necessary and sufficient Karush-Kuhn-Tucker (KKT) conditions.
Let $\partial|t|$ denote the set of subgradients of the absolute value
function $|\cdot|$ at $t$, then it can be shown that
\begin{equation}
s \in \partial |t| \text{ if and only if } t = S(t+s),
\end{equation}
 where $S$ is the soft-thresholding operator with threshold 1, i.e. $S(z)=\text{sgn}(z)(|z|-1)_{+}$.
As shown in Appendix A, combining this fact with some other convex analysis concepts \citep{rockafellar1970convex, combettes2005signal},
 the KKT conditions of \eqref{huber} can be written as
\begin{equation}
\label{KKT1_reformulated}
\begin{cases}
-\frac{1}{n}\sum_{i}h^\prime_{\gamma}(y_{i}-\widehat{\beta}_{0}-x_{i}^{\top}\widehat{\beta}) = 0, \\
-\frac{1}{n}\sum_{i}h^\prime_{\gamma}(y_{i}-\widehat{\beta}_{0}-x_{i}^{\top}\widehat{\beta})x_{ij}+\lambda\alpha \widehat{s}_{j}+\lambda(1-\alpha)\widehat{\beta}_{j} = 0, \\
\widehat{\beta}_{j}-S(\widehat{\beta}_{j}+\widehat{s}_{j})=0, \quad j=1,\ldots, p,
\end{cases}
\end{equation}
where $\widehat{s}_j \in \partial |\widehat{\beta}_j|$ and $h^\prime_{\gamma}(\cdot)$, the derivative of $h_\gamma(\cdot)$, is given by
\[
h^\prime_{\gamma}(t)=\begin{cases}
\frac{t}{\gamma}, & \text{if } |t|\leq\gamma,\\
\text{sgn}(t), & \text{if } |t|>\gamma.
\end{cases}
\]


In this way the optimization problem \eqref{huber} is transformed into a root finding problem for a system of nonsmooth equations \eqref{KKT1_reformulated}. A straightforward approach is applying 
SNA to the entire system of equations. As discussed later in section \ref{subsection_compare}, this approach contains many matrix operations that cause $O(np^2)$ computational cost per iteration, which severely limits its scalability.

For better efficiency and scalability, we propose a new algorithm,  Semismooth Newton Coordinate Descent (SNCD),  that combines SNA with cyclic coordinate descent in solving these equations. Similar to the Gauss-Seidel method for linear equations, SNCD solves the equations of \eqref{KKT1_reformulated} in a cyclic fashion to avoid cumbersome matrix operations. We cycle through $(\beta_0, \beta, s)$ in a pairwise fashion: at each step, a pair $(\beta_j, s_j)$ (and  $\beta_0$ by itself) is updated by solving the corresponding part of \eqref{KKT1_reformulated}, while the other variables are fixed at their current values $\tilde{\beta}_k, \tilde{s}_k, k \neq j$. Specifically, we solve the following equations at each step:

\begin{itemize}
\item For $(\beta_j, s_j)$:
\begin{equation}\label{KKT_bj}
\begin{cases}
-\frac{1}{n}\sum_{i}h^\prime_{\gamma}(\tilde{r}_i + x_{ij}\tilde{\beta}_j- x_{ij}\beta_j)x_{ij}+\lambda\alpha s_{j}+\lambda(1-\alpha)\beta_{j} = 0, \\
\beta_j - S(\beta_j+s_j) = 0,
\end{cases}
\end{equation}
\item For $\beta_0$:
\begin{equation}\label{KKT_b0}
-\frac{1}{n}\sum_{i}h^\prime_{\gamma}(\tilde{r}_i + \tilde{\beta}_0- \beta_0) = 0,
\end{equation}
\end{itemize}
where $\tilde{r}_i = y_i - \tilde{\beta}_0-x_j^\top\tilde{\beta}, \quad i = 1,\ldots, n$.

Note that \eqref{KKT_bj} is exactly the KKT conditions of
\[
\min_{\beta_j}f_H(\ldots, \tilde{\beta}_{j-1}, \beta_j, \tilde{\beta}_{j+1}, \ldots),
\]
and \eqref{KKT_b0} the KKT condition of
\[
\min_{\beta_0} f_H(\beta_0, \tilde{\beta}_{1}, \ldots).
\]
Hence SNCD can be seen as a special type of coordinate descent.

Denote
\begin{equation}\label{psi}
\psi_{\gamma}(t)=\frac{1}{\gamma}I(|t|\leq\gamma),
\end{equation}
then $\psi_\gamma \in \nabla_N h_\gamma^\prime(t), \forall t\in \mathbb{R}$.
The SNCD iterations proceed as follows:

\begin{enumerate}[(i)]
\item
Updating $\beta_0$. Let
\[
F_0(z;\tilde{\beta}) = -\frac{1}{n}\sum_{i}h^\prime_{\gamma}(\tilde{r}_i + \tilde{\beta}_0- z).
\]
Since
\[
H_0(z) = \frac{1}{n}\sum_{i}\psi_\gamma(\tilde{r}_{i}+\tilde{\beta}_{0}-z)\in\nabla_N(F_0(z)),
\]
we update $\beta_0$ by solving \eqref{KKT_b0} via SNA
\[
\beta_{0}\leftarrow \tilde{\beta}_0 - H_0(\tilde{\beta}_0)^{-1}F_0(\tilde{\beta}_0) =
\tilde{\beta}_{0}+
\frac{\sum_{i}h^\prime(\tilde{r}_{i})}{\sum_{i}\psi_\gamma(\tilde{r}_{i})}.
\]

\item
Updating $(\beta_j, s_j)$. Let
\[
F_j(z;\tilde{\beta})=\left[\begin{array}{c}
-\frac{1}{n}\sum_{i} h_\gamma^\prime(\tilde{r}_{i}+x_{ij}\tilde{\beta}_{j}-x_{ij}z_1)x_{ij}+\lambda\alpha z_2+\lambda(1-\alpha)z_1\\
z_1-S(z_1+z_2)
\end{array}\right],
\]
where $z=(z_1, z_2)^{\top}$.
Since
\begin{equation}\label{soft-thresholding decompose}
z_1 - S(z_1 + z_2) =
\begin{cases}
-z_2 + \text{sgn}(z_1+z_2) & \text{if } |z_1 + z_2|>1,\\
z_1 & \text{if } |z_1+z_2|\leq 1,
\end{cases}
\end{equation}
we solve for $(\beta_j, s_j)$ from \eqref{KKT_bj} via SNA in two types of updates:
\begin{enumerate}
\item
$|\tilde{\beta}_{j}+\tilde{s}_{j}|>1$. For $z$ with $|z_1 + z_2|>1$,  a Newton derivative of $F_j$ at $z$ is
\begin{equation}
\label{ND1}
H_{j}(z)=\left[\begin{array}{cc}
\frac{1}{n}
\sum_i\psi_\gamma(\tilde{r}_{i}+
x_{ij}\tilde{\beta}_{j}-x_{ij}z_{1})x_{ij}^{2}+\lambda(1-\alpha) & \lambda\alpha\\
0 & -1
\end{array}\right]\in \nabla_N F_{j}(z).
\end{equation}
Hence the update is
\[
\left[\begin{array}{c} \beta_{j} \\ s_{j} \end{array}\right]
\leftarrow
\left[\begin{array}{c} \tilde{\beta}_{j} \\ \tilde{s}_{j} \end{array}\right]
- H_j(\tilde{\beta}_j, \tilde{s}_j)^{-1} F_j(\tilde{\beta}_j, \tilde{s}_j) =
\left[\begin{array}{c}
\tilde{\beta}_{j}+
\frac{\frac{1}{n}\sum_{i}h_\gamma^\prime(\tilde{r}_{i})x_{ij}-\lambda\alpha \text{sgn}(\tilde{\beta}_{j}+\tilde{s}_{j})-\lambda(1-\alpha)\tilde{\beta}_{j}}
{\frac{1}{n}\sum_{i}\psi_\gamma(\tilde{r}_i)x_{ij}^{2}+\lambda(1-\alpha)} \\
\text{sgn}(\tilde{\beta}_{j}+\tilde{s}_{j})
\end{array}
\right].
\]
\item
$|\tilde{\beta}_{j}+\tilde{s}_{j}|\leq 1$.
For $z$ with $|z_1 + z_2|\leq1$, a Newton derivative of $F_j$ at $z$ is
\begin{equation}
\label{ND2}
H_{j}(z)=\left[\begin{array}{cc}
\frac{1}{n}
\sum_i\psi_\gamma(\tilde{r}_{i}+x_{ij}\tilde{\beta}_{j}-x_{ij}z_{1})x_{ij}^{2}+
\lambda(1-\alpha) & \lambda\alpha, \\
1 & 0
\end{array}\right]\in \nabla_N F_{j}(z).
\end{equation}
Hence the update is
\[
\left[\begin{array}{c} \beta_{j} \\ s_{j} \end{array}\right]
\leftarrow
\left[\begin{array}{c} \tilde{\beta}_{j} \\ \tilde{s}_{j} \end{array}\right]
- H_j(\tilde{\beta}_j, \tilde{s}_j)^{-1} F_j(\tilde{\beta}_j, \tilde{s}_j) =
\left[\begin{array}{c}
0 \\
\frac{
\frac{1}{n}\sum_ih^\prime_\gamma(\tilde{r}_{i})x_{ij}
+\tilde{\beta}_{j}\cdot\frac{1}{n}
\sum_{i}\psi_\gamma(\tilde{r}_{i})x_{ij}^{2}}{\lambda\alpha}
\end{array}
\right].
\]
\end{enumerate}
\end{enumerate}

\subsubsection{Convergence}
Since SNCD fits in the general coordinate descent framework, its convergence property follows
from the convergence results for coordinate descent
\citep{tseng2001convergence}.
To apply the results, we first show that the optimization problem is of the form
\[
\min\; f(z_{1},\ldots,z_{m})=f_{0}(z_{1},\ldots,z_{m})+\sum_{j=1}^m f_{j}(z_{j}),
\]
where $f_0, f_1, \ldots, f_m$ are convex, $f_0$ is first-order differentiable and the level set $\{z: f(z)\leq f(z^0)\}$ is bounded given any initial point $z^0$. A key fact to notice about this formulation is that the nondifferentiable part $\sum_j f_j(z_j)$ must be separable. The penalized Huber loss regression model in \eqref{huber} clearly satisfies these conditions.

At each coordinate update, SNA is applied to solve the equations, which requires nonsingularity of the Newton derivative and the uniform boundedness of its inverse. When updating $\beta_0$, these requirements are met if $|\sum_{i}\psi_\gamma(y_{i}-\beta_{0}-x_{i}^{\top}\beta)|$ is bounded away from 0. This is true as long as there is at least one observation with $|y_i-\beta_0 - x_i^\top \beta|\leq \gamma$. When updating $\beta_j, s_j$, it can be shown via some algebra that a sufficient condition is $0<\alpha<1$ and $\psi_\gamma$ is bounded. The latter always holds since $\psi_\gamma(t) \in \{0,1/\gamma\}$ for any $t$.

In order for this local SNA strategy to work well, we also need the starting point and the optimal point in each coordinate update to be sufficiently close. Denote the globally initial $f_H$ value by $f_H^0$. Since $f_H$ decreases along SNCD iterations and the level set $\{(\beta_0, \beta): f_H(\beta_0, \beta) \leq f_H^0 \}$ is bounded, the closeness requirement is satisfied if the diameter of the set is sufficiently small.

The above discussions are summarized in the following result.

\begin{theorem}\label{SNCD_huber}
For problem \eqref{huber}, let $\lambda > 0$, $\alpha \in (0,1)$ and the initial $f_H$ value be $f_H^0$. Assume for every point $(\beta_0, \beta)$ in the level set $\mathcal{L} = \{(\beta_0, \beta): f_H(\beta_0, \beta) \leq f_H^0 \}$ there exists $i\in\{1,\ldots,n\}$ such that $|y_{i}-\beta_{0}-x_{i}^{\top}\beta|\leq\gamma$. Then SNCD iterations converge to a global minimizer provided that the diameter of $\mathcal{L}$ is sufficiently small.
\end{theorem}

\subsubsection{Pathwise Optimization}\label{path}
To actually implement the algorithm, we still need to consider an important issue: its convergence relies on a good initial point, which is usually not guaranteed in practice. For low-dimensional problems we can use line search to ensure global convergence with an arbitrary initial point, but since line search methods involve considerable amounts of function and gradient evaluations, they are not well-suited for high-dimensional cases.

The strategy of pathwise optimization with warm start can help globalize the convergence of the algorithm. With a decreasing sequence of $\lambda$ values, this strategy sequentially solves the optimization problem at each $\lambda_k$ using the optimizer at the previous $\lambda_{k-1}$ as the initial value. 
When $\lambda_{k-1}$, $\lambda_k$ are reasonably close, the initial point $(\widehat{\beta}_{0}(\lambda_{k-1}),\widehat{\beta}(\lambda_{k-1}))$ will be near the optimizer $(\widehat{\beta}_{0}(\lambda_{k}),\widehat{\beta}(\lambda_k))$ as well.
Hence each optimization problem along the path is warm-started with a good initial point, and fast convergence can be achieved. This strategy generates a solution path, which in turn will be useful for tuning parameter selection.


\subsection{Comparison with SNA and Existing Coordinate Descent Type Algorithms}\label{subsection_compare}
\subsubsection{SNA for Penalized Huber Loss Regression and Its Computational Bottleneck}
\label{subsection SNA}
Denote $\mathcal{S}(z)
=(S(z_{1}),\ldots,S(z_{p}))^\top$ and
$d(\beta_0,\beta)=(h^\prime_{\gamma}(y_{1}-\beta_{0}-x_{1}^{\top}\beta),\ldots,h^\prime_{\gamma}(y_{n}-\beta_{0}-x_{n}^{\top}\beta))^\top$, then the KKT conditions \eqref{KKT1_reformulated} can be written compactly as
\begin{equation} \label{KKT_vector}
F(\beta_0,\beta,s) = \left[
\begin{array}{l}
-\frac{1}{n}1^{\top}d(\beta_0,\beta)\\
-\frac{1}{n}X^{\top}d(\beta_0,\beta)+\lambda\alpha s+\lambda(1-\alpha)\beta\\
\beta-\mathcal{S}(\beta+s)
\end{array}\right]
= \mathbf{0}.
\end{equation}

It is easy to verify $F$ is Newton differentiable, then SNA can be directly applied here for solving $F(\beta_0,\beta,s)=\mathbf{0}$. See Appendix B for details.

In terms of computational cost, the first concern is about matrix inversion, since the Newton derivative of $F$ is a $(1+2p)\times(1+2p)$ matrix, for which inversion becomes intractable when $p$ is large. However, the decomposition \eqref{soft-thresholding decompose} leads to an ``active set strategy'' that helps reduce the dimension. Given the $k$th iteration $(\beta_0^k, \beta^k, s^k)$, define the active set $A_k$ and its complement $B_k$ by
\begin{equation}\label{AB_SNA}
A_k= \{ j:|\beta_{j}^k+s_{j}^k|>1\} \text{ and } B_k= \{j:|\beta_{j}^k+s_{j}^k|\leq 1\}.
\end{equation}
Then the Newton-type iteration of SNA is decomposed into two parts $A_k$ and $B_k$ and only the computation of $\beta_0^{k+1}, \beta_{A_k}^{k+1}$ requires inverting a matrix, the dimension of which is only $(1+|A_k|)\times (1+|A_k|)$. In general, $|A_k|$ can be as large as $p$. But since pathwise optimization is implemented, the algorithm is warm-started at each $\lambda$ value. Hence $A_k$ is usually not too much different from the support of the optimizer, which tends to be a sparse subset of $\{1,\ldots,p\}$.

The real bottleneck is in matrix multiplication. Let $\psi_\gamma$ be as in \eqref{psi}. Let $X^{*}=(\mathbf{1}_{n}\; X)$ and $\Psi_k = \frac{1}{n}\text{diag}(\psi_\gamma(~y_1 -\beta_0^k-x_1^\top\beta^k),\ldots,\psi_\gamma(y_n-\beta_0^k-x_n^\top\beta^k))$.
Then as shown in Appendix B, each iteration includes re-computing and re-partitioning $X^{* \top} \Psi_k X^*$, which involves $O(np^2)$ arithmetic operations that become formidable for large $p$. The diagonality of $\Psi_k$ and the symmetry of $X^{* \top} \Psi_k X^*$ could be utilized to reduce computation, but the magnitude remains $O(np^2)$. Since $X^{*\top}\Psi_{k}X^{*}=
\frac{1}{n}\sum_i\psi_\gamma(y-\beta_{0}^{k}-x_i^{\top}\beta^{k})x_{i}^{*}x_{i}^{*\top}$, caching all the $(1+p)\times (1+p)$ matrices $x_i^{*} x_i^{*\top}$ would also speed up the
computation, but since there are $n$ such matrices, such an implementation would be memory-inefficient.

\subsubsection{SNCD vs. SNA}
The  two algorithms mainly differ in the following aspects:
\begin{itemize}
\item
Consider a full update on $(\beta_0, \beta, s)$ as one iteration. The computational cost per iteration of SNCD is $O(np)$, compared with $O(np^2)$ for SNA.

\item
The SNCD iterations consist of univariate and bivariate updates only while SNA involves 
matrix inversions.

\item
While SNA has locally superlinear convergence rate in theory, SNCD is at most linear. It is a worthwhile compromise, however, considering that SNCD reduces the computational cost per iteration from $O(np^2)$ to $O(np)$ and that warm-starting due to pathwise optimization strategy allows SNCD to converge quickly.

\item
In practice, SNCD is much faster; and SNCD always converges while SNA diverges in some high-dimensional cases even when 
pathwise optimization is used.

\end{itemize}

\subsubsection{SNCD vs. Existing Coordinate Descent Type Algorithms}

SNCD also differs from the existing coordinate descent algorithms for penalized regression (\citealp{friedman2007pathwise}; \citealp{friedman2010regularization}; \citealp{simon2011regularization};
\citealp{breheny2011coordinate})
in the following aspects:
\begin{itemize}
\item
It generalizes coordinate descent to work on a wider class of models where the loss functions, like the Huber loss, only need to be first-order differentiable. As shown in the next section, it is also extended to a case with a nondifferentiable loss, i.e. the quantile loss, via smoothing approximation.
\item
It is directly motivated from the KKT conditions as a root-finding method, where the subgradients $s_j$'s are treated as independent variables that are connected with $\beta_j$'s through the equation $\beta_j - S(\beta_j+s_j) = 0$.
\item
Each pair of $(\beta_j, s_j)$ is updated simultaneously with different formulas for two situations $|\tilde{\beta}_{j}+\tilde{s}_{j}|>1$ and $|\tilde{\beta}_{j}+\tilde{s}_{j}|\leq 1$.
This is quite different from the coordinate descent algorithms mentioned above that only update the coefficients $\beta_j$'s.
\end{itemize}

\section{SNCD for Penalized Quantile Regression}
\label{section_quantile}
\subsection{Description}

For the quantile loss function $\ell = \rho_\tau$, \eqref{gen_loss} becomes
\begin{equation}\label{quantile_enet}
\min_{\beta_{0},\beta} f_Q(\beta_0, \beta) = \frac{1}{n}\sum_{i}\rho_\tau(y_{i}-\beta_{0}-x_{i}^{\top}\beta)+ \lambda P_{\alpha}(\beta).
\end{equation}

SNCD cannot be directly applied to this problem since it requires the first-order derivatives of the loss function, 
but $\rho_\tau$ is not differentiable. However, note that
\[
\rho_{\tau}(t)  =  (1-\tau)t_-+\tau t_+
 =  \frac{1}{2}\left\{ |t|+(2\tau-1)t\right\}.
 \]
Since $h_\gamma(t) \rightarrow |t|$ as $\gamma \rightarrow 0^+$, $\rho_\tau(t) \approx \frac{1}{2}\left\{ h_\gamma(t)+(2\tau-1)t\right\}$ for small $\gamma$ and the solutions to penalized quantile regression can be approximated by
\begin{equation}\label{RA_enet}
\min_{\beta_{0},\beta}\; f_{HA}(\beta_{0},\beta)=\frac{1}{2n}\sum\left\{ h_{\gamma}(y_{i}-\beta_{0}-x_{i}^{\top}\beta)+(2\tau-1)(y_{i}-\beta_{0}-x_{i}^{\top}\beta)
\right\} + \lambda P_{\alpha}(\beta),
\end{equation}
where ``HA" stands for Huber approximation. This problem is easier to handle since its loss function is first-order differentiable.
The following result provides theoretical support for this smoothing approximation.

\begin{theorem}\label{convergence_quantile}
Given any $\lambda \geq 0 $, $0 <\tau < 1$ and $\{\gamma_k\}$ converging to 0, let $(\beta_{0k},\beta_k)$ be a minimizer of $f_{HA}(\beta_0, \beta;\lambda, \tau, \gamma_k)$. Then every cluster point of sequence $\{(\beta_{0k},\beta_k)\}$ is a minimizer of $f_Q(\beta_0, \beta; \lambda, \tau)$.
\end{theorem}

Now we can derive the KKT conditions and apply SNCD to solve \eqref{RA_enet}. Due to its similarity to the penalized Huber loss regression, we omit the details. At each iteration, with the current estimates denoted by $(\tilde{\beta}_0, \tilde{\beta}, \tilde{s})$ and residuals by $\tilde{r}_i$, the SNCD updates are
\begin{enumerate}[(i)]
\item
For $\beta_0$:
\[
\beta_{0}\leftarrow\tilde{\beta}_{0}+\frac{\sum_{i}\left\{ h^\prime_{\gamma}(\tilde{r}_{i})+2\tau-1\right\} }{\sum_{i}\psi_{\gamma}(\tilde{r}_{i})}.
\]
\item
For $(\beta_j, s_j)$:
\begin{enumerate}[(a)]
\item
If $|\tilde{\beta}_{j}+\tilde{s}_{j}|>1$, then
\begin{eqnarray*}
\beta_{j} & \leftarrow & \tilde{\beta}_{j}+\frac{\frac{1}{2n}\sum_{i}\left\{ h^\prime_{\gamma}(\tilde{r}_{i})+2\tau-1\right\} x_{ij}-\lambda\alpha \text{sgn}(\tilde{\beta}_{j}+\tilde{s}_{j})-\lambda(1-\alpha)\tilde{\beta}_{j}}
{\frac{1}{2n}\sum_{i}\psi_{\gamma}(\tilde{r}_i)x_{ij}^{2}+\lambda(1-\alpha)}, \\
s_{j} & \leftarrow & \text{sgn}(\tilde{\beta}_{j}+\tilde{s}_{j}).
\end{eqnarray*}
\item
If $|\tilde{\beta}_{j}+\tilde{s}_{j}|\leq1$, then
\begin{eqnarray*}
\beta_{j} & \leftarrow & 0,\\
s_{j} & \leftarrow & \frac{\frac{1}{2n}\sum_{i}\left\{ h^\prime_{\gamma}(\tilde{r}_{i})+2\tau-1\right\} x_{ij}+\tilde{\beta}_{j}\cdot\frac{1}{2n}
\sum_{i}\psi_{\gamma}(\tilde{r}_{i})x_{ij}^{2}}{\lambda\alpha}.
\end{eqnarray*}
\end{enumerate}
\end{enumerate}

The previous discussions on convergence and pathwise optimization also apply here. And similar to Theorem \ref{SNCD_huber}, we have the following result.
\begin{theorem}\label{SNCD_quant}
For problem \eqref{quantile_enet}, let $\lambda > 0$, $\alpha \in (0,1)$ and the initial $f_{HA}$ value be $f_{HA}^0$.   Assume for every point $(\beta_0, \beta)$ in the level set $\mathcal{L} = \{(\beta_0, \beta): f_{HA}(\beta_0, \beta) \leq f_{HA}^0 \}$ there exists $i\in\{1,\ldots,n\}$ such that $|y_{i}-\beta_{0}-x_{i}^{\top}\beta|\leq\gamma$ . Then SNCD iterations
converge to a global minimizer provided that the diameter of $\mathcal{L}$ is sufficiently small.
\end{theorem}

\subsection{The Choice of $\gamma$ Values}

For the approximation to work well, we need to use a sufficiently small $\gamma$; but when $\gamma$ gets too close to 0, the algorithm becomes ill-conditioned. Therefore we designed a data-dependent heuristic method for picking appropriate $\gamma$ values. At each $\lambda_k$, we determine $\gamma_k$ depending on the residuals $\tilde{r}_i$'s given by the previous optimizer $(\widehat{\beta}_0(\lambda_{k-1}), \widehat{\beta}(\lambda_{k-1}))$
as follows.

\begin{enumerate}[i.]
\item Initialize residuals $\tilde{r}_i \leftarrow y_i$;
\item For each $\lambda_k$:
\begin{enumerate}[(a)]
\item
$\gamma_k \leftarrow \mbox{10-th percentile of } \{|\tilde{r}_i|\}$;
\item
$\gamma_k \leftarrow \min\{\gamma_k, \gamma_{k-1}\}$;
\item
$\gamma_k \leftarrow \max\{\gamma_k, 0.001\}$;
\item
solve the problem with $\gamma_k, \lambda_k$ and update $\tilde{r}_i$'s at each iteration.
\end{enumerate}
\end{enumerate}

In step (a) we  pick a value smaller than the magnitudes of 90\% of all residuals for which the loss function is the same as the quantile loss so the approximation should work well. This also keeps $\gamma_k$ above the magnitudes of 10\% of the residuals, which ensures the numerical stability of the algorithm. Bracketing in (b) and (c) are additional safeguards for stability. 

\subsection{Related Convergence Results}
The key to the smoothing approximation is the fact that $h_\gamma(t)$ converges to $|t|$ as $\gamma$ tends to 0. In fact, it is also easy to see that with $\gamma$ as a scaling factor, $\gamma h_\gamma(t)$ converges to the squared loss $\frac{t^2}{2}$ when $\gamma$ goes to infinity. Hence, in the same spirit of Theorem \ref{convergence_quantile}, we also show the connections between the penalized Huber loss regression and two important regression models with respectively the absolute loss and the squared loss, i.e. the Least Absolute Deviations (LAD) and the Least Squares (LS).

To simplify the notation, let $\theta = (\beta_0, \beta)$ and $P(\cdot)$ be a general penalty function.
Denote
\[
\begin{array}{lll}
\underset{\theta}{\min}\;f_{H}(\theta;\lambda,\gamma) & = & \frac{1}{n}\sum_{i}h_{\gamma}(y_{i}-\beta_{0}-x_{i}^{\top}\beta)+\lambda P(\beta),\\
\\
\underset{\theta}{\min}\;f_{A}(\theta;\lambda) & = & \frac{1}{n}\sum_{i}|y_{i}-\beta_{0}-x_{i}^{\top}\beta|+\lambda P(\beta),\\
\\
\underset{\theta}{\min}\;f_{S}(\theta;\lambda) & = & \frac{1}{2n}\sum_{i}(y_{i}-\beta_{0}-x_{i}^{\top}\beta)^{2}+\lambda P(\beta).
\end{array}
\]

Then we have $f_{H}(\theta;\lambda,\gamma) \rightarrow f_{A}(\theta;\lambda)$ as $\gamma \rightarrow 0$; $\gamma f_{H}(\theta;\lambda/\gamma,\gamma) \rightarrow f_{S}(\theta;\lambda)$ as $\gamma \rightarrow \infty$.
And the following results establish the convergence between their optimizers.
\begin{theorem}\label{convergence_lad}
Given any $\lambda \geq 0$ and $\{\gamma_k\}$ converging to 0, let $\theta_k$ be a minimizer of $f_{H}(\theta;\lambda, \gamma_k)$. Then every cluster point of sequence $\{\theta_k\}$ is a minimizer of $f_A(\theta; \lambda)$.
\end{theorem}

\begin{theorem}\label{convergence_ls}
Given any $\lambda \geq 0$ and $\{\gamma_k\}$ converging to $\infty$, let $\theta_k$ be a minimizer of $f_{H}(\theta;\lambda/\gamma_k,\gamma_k)$. Then every cluster point of sequence $\{\theta_k\}$ is a minimizer of $f_S(\theta; \lambda)$.
\end{theorem}

Therefore, the penalized Huber loss regression bridges the gap between LAD and LS regression as $\gamma$ varies from $0$ to $\infty$.
The solutions of the penalized Huber loss regression constitute a rich spectrum from the solution of LAD regression to that of LS regression. This property gives us more flexibility in fitting high-dimensional regression models.

\section{Adaptive Strong Rule for Screening Predictors}
\label{section_asr}
\cite{tibshirani2012strong} proposed the (sequential) strong rule for screening out predictors in pathwise optimization of penalized regression models for computational efficiency.
However, when applied to the penalized Huber loss regression and quantile regression, we discover that the strong rule suffers from the issue of ``violations" that is explained below. To deal with this issue and enhance algorithmic stability, we develop an adaptive version of the strong rule.

We first describe the strong rule. Consider a general elastic-net penalized regression
\[
\underset{\beta_0, \beta}{\min} \; \frac{1}{n}\sum_{i=1}^n \ell(y_i - \beta_{0} - x_i^\top\beta)+\lambda P_\alpha(\beta).
\]
where $\ell$ is convex and differentiable. Then the optimizer $(\widehat{\beta}_{0}(\lambda),\widehat{\beta}(\lambda))$ satisfies the KKT conditions
\[
\begin{cases}
-\frac{1}{n}\sum_i \ell^\prime(y_i - \widehat{\beta}_{0} - x_i^\top\widehat{\beta})  = 0,\\
-\frac{1}{n}\sum_i \ell^\prime(y_i - \widehat{\beta}_{0} - x_i^\top\widehat{\beta})x_{ij} +\lambda\alpha \widehat{s}_{j}+\lambda(1-\alpha)\widehat{\beta}_{j} = 0,\\
\widehat{s}_{j}\in\partial|\widehat{\beta}_{j}|, \quad j=1,\ldots, p.
\end{cases}
\]

The unpenalized intercept $\beta_0$ is always in the model, so there is no screening rule for it. For $\beta_j$, let $c_{j}(\lambda)=-\frac{1}{n}\sum_i \ell^\prime(y_i - \widehat{\beta}_{0} - x_i^\top\widehat{\beta})x_{ij}$. Assume each $c_j$ is $\alpha$-Lipschitz continuous,
\begin{equation}\label{lip}
|c_{j}(\lambda)-c_{j}(\lambda^{\prime})|\leq\alpha|\lambda-\lambda^{\prime}|,
\ \text{ for every } \ \lambda, \lambda^\prime >0.
\end{equation}
Then at each new $\lambda_k$ in the solution path, given the previous optimizer $(\widehat{\beta}_{0}(\lambda_{k-1}),\widehat{\beta}(\lambda_{k-1}))$ and the corresponding $c_j (\lambda_{k-1})$'s, the strong rule discards predictor $j$ if
\begin{equation}\label{strong}
|c_{j}(\lambda_{k-1})|<\alpha(2\lambda_{k}-\lambda_{k-1}).
\end{equation}
The reasoning is as follows. Assume \eqref{lip} and \eqref{strong} hold, since $\lambda_{k-1} > \lambda_k$, we have
\begin{eqnarray*}
|c_{j}(\lambda_{k})| & \leq & |c_{j}(\lambda_{k})-c_{j}(\lambda_{k-1})|+|c_{j}(\lambda_{k-1})|\\
 & < & \alpha(\lambda_{k-1}-\lambda_{k})+\alpha(2\lambda_{k}-\lambda_{k-1})\\
 & = & \alpha\lambda_{k}.
\end{eqnarray*}
It follows that $\widehat{\beta}_{j}(\lambda_{k})=0$, since by contradiction $\ensuremath{\widehat{\beta}_{j}(\lambda_{k})\neq0}$ implies $ \widehat{s}_{j}(\lambda_k)=\text{sgn}(\widehat{\beta}_{j}(\lambda_{k}))$ thus $|c_{j}(\lambda_{k})|=\lambda_{k}\alpha+\lambda_{k}(1-\alpha)|\widehat{\beta}_{j}(\lambda_{k})|\geq\lambda_{k}\alpha$.

The effectiveness of the strong rule relies on the assumption \eqref{lip}, which does not necessarily hold. So 
application of the rule should always be accompanied with a check of the KKT conditions. A pathwise optimization algorithm incorporating the strong rule proceeds as follows.

For each $\lambda_k$,
\begin{enumerate}[(a)]
\item
Compute the eligible set $E=\{j:\;|c_{j}(\lambda_{k-1})|\geq\alpha(2\lambda_{k}-\lambda_{k-1})\}$;
\item
Solve the problem using only the predictors in $E$;
\item
Check KKT conditions on the solution: $|c_{j}(\lambda_{k})|\leq\alpha\lambda_{k}$ for $j\in E^c$. We are done if there are no violations; otherwise, add violating indices to $E$ and repeat (b) and (c).
\end{enumerate}

For the penalized least squares and logistic regression we have not encountered any violation, but it a different story for the penalized Huber loss regression and quantile regression. Using the strong rule for these two models, we often encounter a large number of violations, indicating that the rule may have been too restrictive.
Since the algorithm is re-run each time violations are found, the overall efficiency is affected. Thus reducing the number of violations can enhance the algorithmic stability and lead to potential speedup.

A simple approach is to use a multiplier $M>1$ and relax the assumption \eqref{lip} to the following: $\forall \lambda, \lambda^\prime >0$,
\[
|c_{j}(\lambda)-c_{j}(\lambda^{\prime})|\leq\alpha M|\lambda-\lambda^{\prime}|.
\]
Accordingly, we will need to change \eqref{strong} to

\[
|c_{j}(\lambda_{k-1})|<\alpha\left(\lambda_{k}+M(\lambda_{k}-\lambda_{k-1})\right).
\]

However, this strategy does not work well in practice, since it is difficult to pre-determine an appropriate value of $M$ that suits all values of $\lambda$ in the solution path.

Hence we propose an ``adaptive" version that allows $M$ to vary with $\lambda$. This rule automatically estimates a localized $M(\lambda)$ that varies and adapts to the trends of the solution paths, which reduces the number of violations by a large margin without sacrificing speed. The idea is as follows.

Let $M(\lambda_0)=1$. Then at each $\lambda_k$,
\begin{enumerate}[(a)]
\item
use $M(\lambda_{k-1})$ to construct the eligible set, i.e. let
\[E=\{j:\;|c_{j}(\lambda_{k-1})|\geq \alpha\left(\lambda_{k}+M(\lambda_{k-1})(\lambda_{k}-\lambda_{k-1})\right)\};
\]
\item
solve the problem using only the predictors in $E$, and check KKT conditions as before; update $E$ and repeat step (b) if violations occur;
\item
compute $M(\lambda_k)$ based on the local trend of $c_j$'s:
\[
M(\lambda_{k})=\frac{\underset{1\leq j\leq p}{\max}|c_{j}(\lambda_{k-1})-c_{j}(\lambda_{k})|}{\alpha(\lambda_{k-1}-\lambda_{k})}.
\]
\end{enumerate}

\section{Numerical results}
\label{section_numerical}
\subsection{Optimization Performance for Penalized Quantile Regression}\label{subsection_opt}

As mentioned in the introduction, \texttt{quantreg}  is another publicly available
R package that supports lasso penalized quantile regression.
Since our implementation employs an approximation model, it does not give ``exact" solutions. Hence we want to compare its solutions with the ones computed by \texttt{quantreg} in terms of optimality.

Unlike \texttt{hqreg} that computes a solution path, \texttt{quantreg} computes a single solution for a given $\lambda$ value, and it does not support the general elastic-net penalty with $0 <\alpha < 1$. For comparison, we only consider lasso ($\alpha = 1$). We first computed a solution path along 100 $\lambda$ values using \texttt{hqreg} and then ran \texttt{quantreg} for each $\lambda$ value. Note that \texttt{quantreg} actually uses the formulation
\[
\min_{\beta_{0},\beta}\sum_{i=1}^{n}\rho_{\tau}(y_{i}-\beta_{0}-x_{i}^{\top}\beta)+\lambda\cdot \frac{1}{2}\sum_{j=1}^p|\beta_{j}|
\]
which does not have a $1/n$ scaling factor for the loss part and instead contains a  $1/2$ factor for the penalty. This is intended to treat the penalty terms as if median regression were performed on them ($\frac{1}{2}\lambda|\beta_j| = \rho_{0.5}(\lambda\beta_j)$). Due to this difference, for each $\lambda$ value used with \texttt{hqreg}, we equivalently supplied \texttt{quantreg} with $2n\lambda$. Also, while \texttt{hqreg} supports data preprocessing via the argument ``preprocess" with 3 options ``standardize", ``rescale" and ``none", \texttt{quantreg} does not provide such an option.
So we standardized the data beforehand for all the real datasets involved in this section and used the standardized ones for comparison. Consequently, we set \texttt{preprocess = "none"} when calling \texttt{hqreg}. For \texttt{quantreg}, the latest version 5.24 was used.

Let $f_Q(\cdot;\lambda)$ denote the objective function as in \eqref{quantile_enet}, and let $\widehat{\beta}_{\text{hqreg}}$ and $\widehat{\beta}_{\text{quantreg}}$ be the solutions given by the two packages, respectively. For $\alpha = 1$ the model is not strictly convex, so in general it does not have a unique optimizer. Hence the values of the two solutions may not be very close. Instead, a reasonable approach is to compare the values of the objective functions $f_Q(\widehat{\beta}_{\text{hqreg}})$ and $f_Q(\widehat{\beta}_{\text{quantreg}})$. Specifically, we made the comparisons based on the relative difference,
\begin{equation}
\label{Ddef}
D(\lambda) = \frac{f_{Q}(\widehat{\beta}_{\text{hqreg}};\lambda)-f_{Q}(\widehat{\beta}_{\text{quantreg}};\lambda)}{f_{Q}(\widehat{\beta}_{\text{quantreg}};\lambda)}.
\end{equation}

Two datasets were considered:
\begin{itemize}
\item
GDP \citep{koenker1999goodness}: consists of 161 observations on national GDP growth rates, recorded as ``Annual Change Per Capita GDP", and 13 covariates. The first 71 observations are from the period 1965-1975, and the rest from the period 1975-1985. This dataset is available in \texttt{quantreg} via \texttt{data(barro)}.

\item
Riboflavin \citep{buhlmann2014high}: gene-expression data for predicting log transformed riboflavin (vitamin B2) production rate in Bacillus subtilis. It contains 71 observations and 4088 features (gene expressions). This dataset is available in R package \texttt{hdi} via \texttt{data(riboflavin)}. For this task only 1000 features with the largest variances were used.
\end{itemize}

\begin{figure}[htbp]
\begin{center}
\includegraphics[scale=0.4]{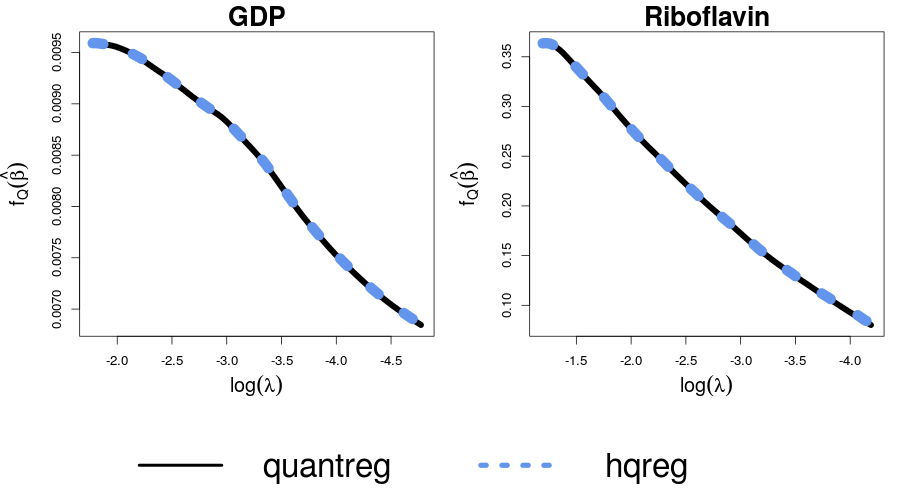}
\vspace{-7mm}
\caption{Values of objective functions with $\tau=0.5$ along the solution path for GDP and riboflavin datasets. Solid line: \texttt{quantreg}, dashed line: \texttt{hqreg}.}
\label{plot_objective_quant}
\end{center}
\end{figure}

\begin{table}[htbp]
\spacingset{1}
\begin{center}
\begin{tabular}{c|c|rr}
\hline
Dataset & $\tau$ & $\min D(\lambda_i)$
& $\max D(\lambda_i)$ \\
\hline
\multirow{3}{*}{GDP} &0.25 & -2.1e-9 & 1.5e-3 \\
&0.50 & -1.3e-10 & 9.6e-4 \\
&0.75 & -3.0e-10 & 1.7e-3 \\
\hline
\multirow{3}{*}{Riboflavin} &0.25 & 6.5e-5 & 2.6e-2 \\
&0.50 & -3.6e-10 & 2.0e-2 \\
&0.75 & 8.8e-5 & 2.1e-2 \\
\hline
\end{tabular}
\caption{
The range of the relative differences $D(\lambda_i), 1 \le i \le 100,$ between \texttt{hqreg} and \texttt{quantreg}.}
\label{compare_optimality_quant}
\end{center}
\end{table}

\begin{table}[htbp]
\spacingset{1}
\begin{center}
\begin{tabular}{c|c|c|ccc}
\hline
\multirow{2}{*}{Dataset} & \multirow{2}{*}{$\tau$} & \multirow{2}{*}{\texttt{hqreg}} & \multicolumn{3}{c}{\texttt{quantreg}} \\
& &  & total & $\lambda_1$ & $\lambda_{100}$ \\
\hline
\multirow{3}{*}{GDP} &0.25 & 0.018 & 0.235 & 0.007 & 0.002\\
&0.50 & 0.018 & 0.223 & 0.002  & 0.003\\
&0.75 & 0.027 & 0.240 & 0.003 & 0.002\\
\hline
\multirow{3}{*}{Riboflavin} &0.25 & 2.501 & 538.2 & 3.630 & 4.958 \\
&0.50 & 3.026 & 531.6 & 4.591  & 4.984\\
&0.75 & 2.922 & 588.8 & 7.119 & 5.791\\
\hline
\end{tabular}
\caption{Running time (in seconds) for computing the solution paths}
\label{compare_time_quant}
\end{center}
\end{table}

Figure \ref{plot_objective_quant} displays the computed values of objective functions $f_Q(\widehat{\beta}_{\text{hqreg}})$ and $f_Q(\widehat{\beta}_{\text{quantreg}})$ for $\tau = 0.5$ along the solution path for both datasets. There is no visually detectable discrepancy between the two lines. Hence we also computed the range of $D(\lambda)$ in each case and the results are listed in Table \ref{compare_optimality_quant}. In each case, the range of $D(\lambda_i)$'s is extremely narrow and all values are very close to zero. This indicates the two packages indeed have similar performances.

We also report the running time in Table \ref{compare_time_quant}. The time for \texttt{hqreg} is for one call that fits the entire solution path, and the time for \texttt{quantreg} is the total of time recorded separately for each $\lambda$. For all these cases \texttt{hqreg} is significantly faster than \texttt{quantreg}, although it may not be quite fair for \texttt{quantreg} since it does not rely on warm-start. The timings taken for \texttt{quantreg} on $\lambda_1$ and $\lambda_{100}$ are also listed, which appear to be roughly the same. In the case of riboflavin data, the running time of \texttt{quantreg} on single $\lambda$ values is in fact longer than the time used by \texttt{hqreg} to compute the whole path.

To further investigate their performances in various other scenarios, we ran a large set of experiments on 10000 datasets, each generated with the following settings:
\begin{itemize}
\item
the number of observations $n$ and the number of features $p$ are randomly selected from the set $\{20,100,200,500,1000,2000,5000\}$.
\item
the number of nonzero coefficients is $q = \theta \min(n, p)$ where $\theta$ is uniformly sampled from $\{5\%, 10\%, 20\%, 30\%\}$ and the coefficients values are randomly selected from $\{\pm 1, \ldots, \pm 10 \}$.
\item
each feature vector $x_i$ is generated via $x_{ij} = z_{ij} + 0.5 u_i, \; 1\leq j \leq p$, where $z_{ij}, u_i$'s are i.i.d. standard gaussian, so that each pair of features has the same correlation 0.25.
\item
the outcome $y_i$'s are generated by $y_i = 10 + x_i^\top\beta+\veps_i$, where
$\veps_i$'s are iid sampled from Student's t distribution with $df = 4$.
\end{itemize}

For each dataset and each $\tau \in \{0.25, 0.5, 0.75\}$, we applied \texttt{hqreg} to compute an entire solution path and randomly selected an index $k$ out of $\{10, 20, \ldots, 100\}$, then ran \texttt{quantreg} on $\lambda_k$, the k-th $\lambda$ value for the solution path computed by \texttt{hqreg}. These experiments were performed in parallel via grid computing on a high performance cluster at the University of Iowa.

\begin{figure}[htbp]
\begin{center}
\includegraphics[scale=0.4]{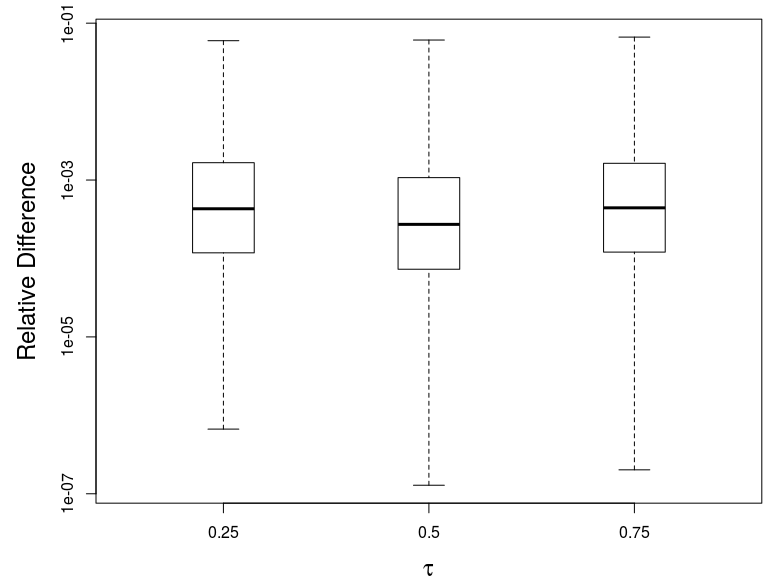}
\vspace{-5mm}
\caption{Boxplots of the relative difference $D$ on 10000 simulated datasets}
\label{box_objective}
\end{center}
\end{figure}

We calculate the relative difference $D(\lambda)$ for each pair of the solutions and summarize the results in three boxplots plotted on the logarithmic scale shown in Figure \ref{box_objective}. We observe that the values of $D$ have a narrow range between  1e-7 and 1e-1 with the majority falling below 1e-3, and are slightly smaller for $\tau = 0.5$. Besides, the distribution of $D$ appears roughly symmetric on the logarithmic scale in each case.


\subsection{Timing Performance}
\label{subsection_timings}

In addition to the Huber loss and the quantile loss, \texttt{hqreg} also supports the squared loss for the least squares which is not discussed in this paper, but its SNCD iterations can be derived in a similar way as the other two models. Here we consider their running time performances.

We generated Gaussian data with $n$ observations and $p$ features, where each pair of features have an identical correlation $\rho$. To simplify settings and highlight the timing comparison based on the key parameter $\gamma$ and $\tau$, we set $\rho = 0.25$ and $\alpha = 0.9$ for all cases. The responses were generated by
\[
Y=\sum_j X_{j}\beta_{j}+k\cdot E
\]
where $\beta_{j}=(-1)^{j}\exp(-(j-1)/10)$, $E\sim T(df=4)$ and $k$ is determined so that the signal-to-noise ratio is 3.

\subsubsection{Huber Loss Regression and Least Squares}
In this part, we compare the running time of competing methods for the elastic-net penalized Huber loss regression and least squares. For the Huber loss,  since there is no other algorithms, we consider only \texttt{hqreg} for SNCD with no variable screening (\texttt{hqreg}-NVS), SNCD with the adaptive strong rule (\texttt{hqreg}-ASR),
and our implementation of pure SNA.
In the experiments we considered 5 values of $\gamma$ ranging from 0.01 to 100.  On the other hand, for the least squares we compared {\texttt{hqreg}} with a state-of-the-art coordinate descent algorithm implemented by R package \texttt{glmnet}. For \texttt{glmnet}, the latest version 2.0-5 was used which employs the strong rule for variable screening. All methods considered here are R functions. \texttt{glmnet} does all its computation in Fortran, \texttt{hqreg} does the computation in C, and the SNA implementation is also programmed in C with matrix operations performed via BLAS and LAPACK.

\begin{table}[htbp]
\spacingset{1}
\begin{center}
\begin{tabular}{lrrrrr|c}
& \multicolumn{5}{c|}{Huber} & Least Squares \\
& \multicolumn{5}{c|}{$\gamma$} & \\
& 0.01 & 0.1 & 1 & 10 & 100 & \\
\cline{2-7}\\[0.01ex]
& \multicolumn{6}{c}{$n = 1000, p = 100$} \\
\cline{2-7}
hqreg-NVS &0.61 &0.09 &0.05 &0.04 &0.04 & 0.03\\
hqreg-ASR &0.33 &0.06 &0.03 &0.02 &0.02 & 0.02\\
glmnet & --- &--- &--- &--- &--- & 0.02 \\

\cline{2-7}\\[0.01ex]
& \multicolumn{6}{c}{$n = 5000, p = 100$} \\
\cline{2-7}
hqreg-NVS &2.46 &0.48 &0.24 &0.20 &0.21 & 0.14 \\
hqreg-ASR &1.32 &0.30 &0.15 &0.12 & 0.11 &0.07 \\
glmnet & --- &--- &--- &--- &---  & 0.02 \\
\cline{2-7}\\[0.01ex]
& \multicolumn{6}{c}{$n = 100, p = 1000$} \\
\cline{2-7}
hqreg-NVS & 1.89 &0.39 &0.09 &0.08 &0.08 &0.05 \\
hqreg-ASR &0.52 &0.11 &0.03 &0.02 &0.02 & 0.02 \\
glmnet & --- &--- &--- &--- &---  & 0.02 \\
\cline{2-7}\\[0.01ex]
& \multicolumn{6}{c}{$n = 100, p = 5000$} \\
\cline{2-7}
hqreg-NVS &8.70 &2.09 &0.46 &0.38 &0.38 & 0.29\\
hqreg-ASR &0.85 &0.23 &0.09 &0.11 &0.08 & 0.07 \\
glmnet & --- &--- &--- &--- &---  & 0.08 \\
\cline{2-7}\\[0.01ex]
& \multicolumn{6}{c}{$n = 100, p = 20000$} \\
\cline{2-7}
hqreg-NVS & 30.27 & 8.88 &2.43 & 2.45 & 2.40 & 1.23\\
hqreg-ASR &1.60 &0.54 &0.43 &0.31 &0.30 & 0.32 \\
glmnet & --- &--- &--- &--- &--- & 0.30 \\
\cline{2-7}\\[0.01ex]
& \multicolumn{6}{c}{$n = 100, p = 100000$} \\
\cline{2-7}
hqreg-NVS &175.81 &45.33 & 11.23 & 11.49 & 11.41  & 5.94 \\
hqreg-ASR &4.50 &2.12 &1.69 &1.58 & 1.53 & 1.57 \\
glmnet & --- &--- &--- &--- &--- & 1.39 \\
\cline{2-7}
\end{tabular}
\caption{Running time (in seconds) for computing regularization paths for the elastic-net penalized Huber loss regression and least squares regression. Total time for \texttt{100} $\lambda$ values, averaged over 3 runs. 
}
\label{timings}
\end{center}
\end{table}

We have found in practice that convergence of SNA has much higher reliance on initial points than SNCD, and it can fail if the $\lambda$ sequence is not dense enough. Hence we divided the experiments into two parts. In the first part for the Huber loss and the least squares together, we left out SNA and computed each solution path with the usual number of 100 $\lambda$ values. In the second part, we compared only SNA and SNCD(\texttt{hqreg}-NVS) for the Huber loss on dense lambda sequences each consisting of 10000 values.

Table \ref{timings} shows average CPU timings for the first part. First compare the timings for the Huber loss. Across different values of $\gamma$, we observe that for both versions the timings increase when $\gamma$ is nearing 0, and stay almost the same for $\gamma \geq 1$. And clearly \texttt{hqreg}-ASR that employs the adaptive strong rule is much faster and more scalable than \texttt{hqreg}-NVS that has no variable screening. For the least squares, \texttt{hqreg}-ASR and \texttt{glmnet} have similar performances except the case with $n = 5000$. Besides, we discover that the timings for the Huber loss regression with $\gamma \geq 1$ are very close to those of the least squares. Considering that the Huber loss is more difficult to handle than the simple squared loss, the performance of \texttt{hqreg} is very impressive.

Table \ref{timings_small} shows average CPU timings for the second part. 
We observe that while SNCD converges in every case, SNA fails in the cases with large $p$ and $\gamma = 0.1$. When $p$ is small, SNCD does not have much advantage. But when $p$ increases, SNCD becomes considerably faster with an increasing speedup relative to SNA. These results show that SNCD is more stable and scalable than SNA.

\begin{table}[htbp]
\spacingset{1}
\begin{center}
\begin{tabular}{lrrr}
& \multicolumn{3}{c}{$\gamma$} \\
& 0.1 & 1 & 10 \\
\cline{2-4}\\[0.01ex]
& \multicolumn{3}{c}{$n = 1000, p = 100$} \\
\cline{2-4}
SNA & 3.98 & 5.16 & 5.44 \\
SNCD & 1.89 & 1.27 & 1.19 \\
\cline{2-4}\\[0.01ex]
& \multicolumn{3}{c}{$n = 5000, p = 100$} \\
\cline{2-4}
SNA & 17.98 & 24.42 & 26.33 \\
SNCD & 12.38 & 6.84 & 6.31 \\
\cline{2-4}\\[0.01ex]
& \multicolumn{3}{c}{$n = 100, p = 1000$} \\
\cline{2-4}
SNA & $\times$ & 11.70 & 10.47 \\
SNCD & 2.24 & 1.62 & 1.51 \\
\cline{2-4}\\[0.01ex]
& \multicolumn{3}{c}{$n= 100, p = 5000$} \\
\cline{2-4}
SNA & $\times$ & 98.66 & 100.76 \\
SNCD & 9.87 & 8.19 & 8.96 \\
\cline{2-4}\\[0.01ex]
\end{tabular}
\vspace{-5mm}
\caption{Running time (in seconds) for comparing SNCD(\texttt{hqreg}-NVS) and SNA on the penalized Huber loss regression. ``$\times$" represents early exit due to divergence at some $\lambda$ value. Total time for \texttt{10000} $\lambda$ values. }
\label{timings_small}
\end{center}
\end{table}

\subsubsection{Quantile Regression}

\texttt{hqreg} is faster than \texttt{quantreg} for the examples in section \ref{subsection_opt}. However, \texttt{quantreg} does not implement pathwise optimization and rely on warm-start like \texttt{hqreg} does. Instead, for each supplied $\lambda$ value it has to solve the corresponding problem individually ``from scratch". So it is not quite reasonable to compare \texttt{quantreg} with \texttt{hqreg} for computing the whole solution path. For this part, we compare only \texttt{hqreg}-NVS and \texttt{hqreg}-ASR. As shown in Table \ref{timings_quantile}, \texttt{hqreg}-ASR is similar to \texttt{hqreg}-NVS in cases with $p=100$ but considerably faster when $p$ gets larger. \texttt{hqreg}-ASR also shows much better scalability with the dimension $p$.

\begin{table}[htb]
\spacingset{1}
\begin{center}
\begin{tabular}{lrrr}
& \multicolumn{3}{c}{$\tau$} \\
& 0.25 & 0.50 & 0.75 \\
\cline{2-4}\\[0.01ex]
& \multicolumn{3}{c}{$n = 1000, p = 100$} \\
\cline{2-4}
hqreg-NVS & 0.21 & 0.18 & 0.19 \\
hqreg-ASR & 0.13 & 0.10 & 0.11 \\
\cline{2-4}\\[0.01ex]
& \multicolumn{3}{c}{$n = 5000, p = 100$} \\
\cline{2-4}
hqreg-NVS & 0.56 & 0.58 & 0.54 \\
hqreg-ASR & 0.38 & 0.42 & 0.33 \\
\cline{2-4}\\[0.01ex]
& \multicolumn{3}{c}{$n = 100, p = 1000$} \\
\cline{2-4}
hqreg-NVS & 10.77 & 7.37 & 11.90 \\
hqreg-ASR & 2.98 & 1.94 & 2.92 \\
\cline{2-4}\\[0.01ex]
& \multicolumn{3}{c}{$n = 100, p = 5000$} \\
\cline{2-4}
hqreg-NVS & 47.08 & 41.46 & 58.97 \\
hqreg-ASR & 3.33 & 2.92 & 4.23 \\
\cline{2-4}\\[0.01ex]
& \multicolumn{3}{c}{$n = 100, p = 100000$} \\
\cline{2-4}
hqreg-ASR & 19.28 & 12.43 & 22.98 \\
\cline{2-4}\\[0.01ex]
\end{tabular}
\vspace{-5mm}
\caption{Running time (in seconds) for computing regularization paths for penalized quantile regression. Total time for \texttt{100} $\lambda$ values, averaged over 3 runs.}
\label{timings_quantile}
\end{center}
\end{table}

\subsection{Real Data Example}
We now compare the modelling performance of penalized Huber loss regression, quantile regression and least squares via an empirical analysis on a real dataset. It is a breast cancer gene expressions dataset that comes from the Cancer Genome Atlas (2012) project (\url{http://cancergenome.nih.gov/}), obtained using Agilent mRNA expression microarrays. It contains expression measurements of 17814 genes on 536 patients, including BRCA1, the first gene identified to be associated with increasing risk of early onset breast cancer. Hence we regress the key gene BRCA1 on the other genes to detect potential interconnections. Before fitting the models, we carried out the following two screening steps: remove any gene for which the range of the expression among all patients is less than 2, and remove any gene for which the sample correlation with BRCA1 is less than 0.05. After the screening, there are 11562 genes left.
Then we consider 7 elastic-net penalized linear regression models using these genes as predictors: the least squares (LS-Enet); 3 Huber loss regression models with values of $\gamma$ being $\text{IQR}(y)$, $\text{IQR}(y)/2$, $\text{IQR}(y)/10$ respectively where $\text{IQR}(y) = 0.93$, denoted as H-Enet($\gamma = \text{IQR}(y)$), H-Enet($\gamma = \text{IQR}(y)/2$), and (H-Enet($\gamma = \text{IQR}(y)/10$); 3 quantile regression models with $\tau = 0.25, 0.50, 0.75$, denoted as Q-Enet($\tau = 0.25$), Q-Enet($\tau = 0.50$), and Q-Enet($\tau = 0.75$). $\alpha = 0.9$ is applied to the elastic-net penalty for all models.

We conduct 50 random partitions. For each partition, we randomly select 300 patients as the training data and the other 236 as the testing data. A five-fold cross validation is applied to the training data to select the tuning parameter $\lambda$. For prediction on the testing set, we consider two error measures. The first one is the commonly used mean absolute prediction error (MAPE). Since MAPE is not sensitive to heterogeneity and may not provide accurate assessment for Q-Enet($\tau = 0.25$) and Q-Enet($\tau = 0.75$) which use asymmetric losses, we also consider using the quantile loss $\rho_\tau$ to measure prediction performance as suggested in \cite{wang2012quantile}. With $\rho_\tau$ for corresponding quantile regression models and $\rho_{0.5}$ for the least squares and the Huber loss regression models, we define quantile-based prediction error (QPE) as $\sum_{i} \rho_\tau(y_i - \widehat{y}_i)/n$.

\begin{table}[htbp]
\spacingset{1}
\begin{center}
\begin{tabular}{lrrr}
\multicolumn{1}{c}{Model} & \multicolumn{1}{c}{Ave \# nonzero} & \multicolumn{1}{c}{Ave MAPE} & \multicolumn{1}{c}{Ave QPE}  \\
\hline
LS-Enet & 114.30 (36.99) & 0.335 (0.018) & 0.167 (0.009) \\
H-Enet($\gamma = \text{IQR}(y)$) & 100.14 (44.70) & 0.331 (0.018) & 0.166 (0.009) \\
H-Enet($\gamma = \text{IQR}(y)/2$) & 82.06 (30.40) &
0.310 (0.020) & 0.155 (0.010) \\
H-Enet($\gamma = \text{IQR}(y)/10$) & 114.08 (30.73) & 0.293 (0.021) & 0.146 (0.010) \\
Q-Enet($\tau = 0.25$) & 94.58 (41.60) & 0.373 (0.026) & 0.151 (0.010) \\
Q-Enet($\tau = 0.50$) & 152.90 (51.96) & 0.294 (0.021) & 0.147 (0.012) \\
Q-Enet($\tau = 0.75$) & 104.90 (27.96) & 0.317 (0.027) & 0.109 (0.007) \\
\hline
\end{tabular}
\caption{Analysis of the microarray dataset}
\label{bcdata_comparison}
\end{center}
\end{table}

In Table \ref{bcdata_comparison} we report the average number of nonzero regression coefficients, the average MAPE and the average QPE, where numbers in the parentheses are the corresponding standard errors  across the 50 partitions. The standard errors of the estimated numbers of nonzero coefficients are large relative to the averages, showing that all models are affected by noise to some extent. However, the standard errors for MAPE and QPE are relatively small, which indicates the prediction performances are stable. Among all models, H-Enet($\gamma = \text{IQR}(y)/10$) and Q-Enet($\tau = 0.50$) have the best performances in terms of MAPE, and Q-Enet($\tau = 0.75$) dominates QPE, while LS-Enet performs poorly under both criteria. Q-Enet($\tau = 0.75$) seems the best overall and it also tends to select sparser models compared to the aforementioned H-Enet($\gamma = \text{IQR}(y)/10$), Q-Enet($\tau = 0.50$) or LS-Enet.

For each model, different partitions may lead to different selection results.
We select LS-Enet, H-Enet($\gamma = \text{IQR}(y)/10$) and Q-Enet($\tau = 0.75$) to represent their own classes, and report the names and the frequencies of top genes selected (over 40 times) in Table \ref{freq_bcdata} where the genes are ordered alphabetically. We observe that some genes such as DTL, NBR2, PSME3, RPL27 have high frequencies with all three models, while genes such as KHDRBS1 do not. Overall, H-Enet($\gamma = \text{IQR}(y)/10$) and Q-Enet($\tau = 0.75$) select more genes with high frequencies than LS-Enet while their model sizes are smaller on average, especially Q-Enet($\tau = 0.75$). It indicates these two models more consistently capture the important genes.

\begin{table}[htbp]
\spacingset{1}
\begin{center}
\begin{tabular}{lclclc}
\multicolumn{2}{c}{LS-Enet} & \multicolumn{2}{c}{H-Enet($\gamma = \text{IQR}(y)/10$)} & \multicolumn{2}{c}{Q-Enet($\tau = 0.75$)} \\
\hline
Gene & \multicolumn{1}{c}{Frequency} &
Gene & \multicolumn{1}{c}{Frequency} &
Gene & \multicolumn{1}{c}{Frequency} \\
\hline
DTL & 45 & C17orf53 & 46 & C17orf53 & 48 \\
KHDRBS1 & 41 & CENPQ & 42 & CENPM & 45 \\
NBR2 & 50 & DTL & 46 & DTL & 44 \\
PSME3 & 45 & MCM6 & 50 & GCN5L2 & 44 \\
RPL27 & 45 & NBR1 & 47 & KIAA0101 & 40 \\
VPS25 & 43 & NBR2 & 50 & MCM6 & 42 \\
& & NMT1 & 41 & NBR1 & 49 \\
& & PSME3 & 50 & NBR2 & 50 \\
& & RPL27 & 41 & PSME3 & 50 \\
& & & & RPL27 & 50 \\
& & & & SUZ12 & 40 \\
& & & & SYNGR4 & 41 \\
& & & & XRCC2 & 41 \\
\hline
\end{tabular}
\caption{Genes selected with high frequency for the microarray dataset}
\label{freq_bcdata}
\end{center}
\end{table}

\section{Discussions}
\label{section_discuss}
The Huber loss regression and the quantile regression have important applications in many fields.
However, there is a lack of efficient algorithms and publicly available software that can fit these models in high-dimensional settings. In this paper, we develop an efficient and scalable algorithm for computing the solution paths for these models
with the elastic-net penalty. We also provide an implementation via the R package \texttt{hqreg} publicly available on CRAN (\url{http://cloud.r-project.org/package=hqreg}).

\appendix
\section*{Appendices}
\renewcommand{\thesubsection}{\Alph{subsection}}
\counterwithin{theorem}{subsection}
\counterwithin{equation}{subsection}

\subsection{Background on Convex Analysis and Properties of Newton Derivative}
To derive the KKT conditions \eqref{KKT1_reformulated}, we recall some background in convex analysis. We also describe some useful properties of Newton derivative.

For a convex function $f$, a vector $w$ is called a \emph{subgradient} of $f$ at $z$ if 
\begin{equation}
f(x)-f(z)\geq w^{\top}(x-z),\quad\forall x.
\end{equation}
The set of all subgradients of $f$ at $z$ is called the \emph{subdifferential}, denoted as $\partial f(z)$. For example, the subdifferential of the absolute value function has the following form

\begin{equation}
\partial|z|=\begin{cases}
\{\text{sign}(z)\} & \text{if } z \neq 0, \\
[-1,1] & \text{if } z=0.
\end{cases}
\end{equation}

For convex optimization problems, the necessary and sufficient optimality conditions are called the KKT conditions. In the case of unconstrained optimization, the KKT conditions can be stated in terms of Fermat's rule \citep{rockafellar1970convex}: for a convex function $f$,
\begin{equation}\label{Fermat}
\mathbf{0} \in\partial f(z^{*})\Leftrightarrow z^{*}=\arg\underset{z}{\min}f(z).
\end{equation}

This holds because by definition $\mathbf{0} \in\partial f(z^{*})$ if and only if $f(z)-f(z^*)\geq \mathbf{0}^\top (z-z^*)=0$ for every $z$, that is, $z^{*}=\arg\underset{z}{\min}f(z)$.

A more general result \citep{combettes2005signal} is
\begin{equation}\label{genFermat}
w\in\partial f(z)\Leftrightarrow z=\text{Prox}_{f}(z+w),
\end{equation}
where $\text{Prox}_f$ is the \emph{proximity operator} for $f$ defined as
\[
\text{Prox}_{f}(z)\coloneqq \arg\underset{x}{\min}\frac{1}{2}\|x-z\|_2^{2}+f(x).
\]
The second statement can be shown as follows. Applying Fermat's rule,
\[
z=\text{Prox}_{f}(z+w)=\arg\underset{x}{\min}\frac{1}{2}\|x-z-w\|_2^{2}+f(x),
\]
if and only if there exists $s \in \partial f(x)$ such that
\[
0=(z-z-w)+s=-w+s,
\]
that is,
\[
w=s\in\partial f(x).
\]

It can shown that the proximity operator of the absolute value $|\cdot|$ is given in closed form by the soft-thresholding operator with threshold 1, i.e.

\begin{equation}
\text{Prox}_{|\cdot|}(z)= S(z)=\text{sgn}(z)(|z|-1)_{+}.
\end{equation}

Then it follows from \eqref{genFermat} that $s_j \in \partial|\beta_j|$ can be expressed as an equation
\begin{equation}\label{subgradient_eq}
\beta_{j}-S(\beta_{j}+s_{j})=0.
\end{equation}


According to the Fermat's rule \eqref{Fermat}, the KKT conditions for the penalized Huber loss regression \eqref{huber} are

\begin{equation}
\begin{cases}
-\frac{1}{n}\sum_{i}h^\prime_{\gamma}(y_{i}-\widehat{\beta}_{0}-x_{i}^{\top}\widehat{\beta}) = 0, \\
-\frac{1}{n}\sum_{i}h^\prime_{\gamma}(y_{i}-\widehat{\beta}_{0}-x_{i}^{\top}\widehat{\beta})x_{ij}+\lambda\alpha \widehat{s}_{j}+\lambda(1-\alpha)\widehat{\beta}_{j} = 0, \\
\widehat{s}_j \in \partial|\widehat{\beta}_j|, \quad j=1,\ldots, p,
\end{cases}
\end{equation}
where $(\widehat{\beta}_0, \widehat{\beta})$ is an optimizer. Rewriting the last row by \eqref{subgradient_eq}, we obtain the KKT conditions as a system of equations \eqref{KKT1_reformulated}.

The definition of ``Newton derivative" is already given in the main text. Now we provide several properties useful for calculating Newton derivatives. The first one is the following chain rule for Newton derivatives \citep{ito2008lagrange}.

\begin{lemma}\label{chain}
If $F:\mathbb{R}^{l}\rightarrow\mathbb{R}^{m}$ is continuously Fr\'echet differentiable at $z\in \mathbb{R}^l$ with Jacobian $J_F$ and $G:\mathbb{R}^{m}\rightarrow\mathbb{R}^{n}$ is Newton differentiable at $F(z)$ with a Newton derivative $H_G$. Then $G\circ F$ is Newton differentiable at $z$ with a Newton derivative $H_G(F(z+h))J_F(z+h)$ for $h$ sufficiently small.
\end{lemma}

We also provide two other results.

\begin{lemma}\label{basic}
In the following, assume $F:\mathbb{R}^{m}\rightarrow\mathbb{R}^{l}$, $G:\mathbb{R}^{m}\rightarrow\mathbb{R}^{l}$, $z\in\mathbb{R}^{m}$, $F=(F_{1},\ldots,F_{l})^{\top}$ and $H = (H_1^\top, \ldots, H_l^\top)^\top$, where $H_i \in \mathbb{R}^{1\times m}, \quad i = 1,\ldots,l$.
~\begin{enumerate}[(i)]
\item
If $F$ is continuously Fr\'echet differentiable at $z$, then $F$ is also Newton differentiable at $z$ and $ J_F \in\nabla_N F(z)$;
\item
If $F$ is Newton differentiable at $z$, then for any integer $k>0$ and $A\in\mathbb{R}^{k\times l}$, $AF$ is Newton differentiable at $z$; if $H\in\nabla_NF(z)$, then $AH\in\nabla_NAF(z)$;
\item
If $F$ and $G$ are Newton differentiable at $z$, then $F+G$ is Newton differentiable at $z$; if $H_F\in\nabla_NF(z)$, $H_G\in \nabla_NG(z)$, then $H_{F}+H_{G}\in\nabla_N(F+G)(z)$;
\item
$F$ is Newton differentiable at $z$ if and only if $F_1,\ldots,F_l$ are all Newton differentiable at $z$ and $H\in\nabla_NF(z)\Leftrightarrow H_{i}\in\nabla_NF_{i}(z),\quad i=1,\ldots,l$;
\end{enumerate}
\end{lemma}

\begin{lemma}\label{piecewise}
A univariate piecewise-smooth real function $f$ is everywhere Newton differentiable, with a Newton derivative $H$ given by
\[
H(z)=\begin{cases}
f'(z) & \text{if } f\text{ is differentiable at }z, \\
r_{z}\in\mathbb{R}^{1} & \text{if } f\text{ is not differentiable at }z.
\end{cases}
\]
\end{lemma}

\subsection{Derivation of SNA for Penalized Huber Loss Regression}

Following section \ref{subsection SNA}, denote $\mathcal{S}(z)
=(S(z_{1}),\ldots,S(z_{p}))^\top$ and
$d(\beta_0,\beta)=(h^\prime_{\gamma}(y_{1}-\beta_{0}-x_{1}^{\top}\beta),\ldots,h^\prime_{\gamma}(y_{n}-\beta_{0}-x_{n}^{\top}\beta))^\top$, then the KKT conditions \eqref{KKT1_reformulated} can be written as \eqref{KKT_vector}.

Since the soft-thresholding operator is piecewise linear as shown in \eqref{soft-thresholding decompose}, we define

\begin{equation}\label{ABsets}
\begin{aligned}
&A = \left\{ j:|\beta_{j}+s_{j}|>1\right\} ,\\
&B = \{j:|\beta_{j}+s_{j}|\leq 1\}.
\end{aligned}
\end{equation}

The set $A$ works as an estimate for the support of $\beta$. In fact, if $(\widehat{s}, \widehat{\beta}_0, \widehat{\beta})$ satisfies the KKT conditions, then the set $A$ defined on $(\widehat{\beta}, \widehat{s})$ is exactly the support for $\widehat{\beta}$. This is easy to see: since $\widehat{s}_j \in \partial |\widehat{\beta}_j|$, if $\widehat{\beta}_j \neq 0$ then $|\widehat{\beta}_j+\widehat{s}_j| = |\widehat{\beta}_j+\text{sgn}(\widehat{\beta}_j)| = |\widehat{\beta}_j|+1>1$, otherwise $|\widehat{\beta}_j+\widehat{s}_j| = |\widehat{s}_j| \leq 1$. 

We decompose $\beta$ into $\beta_A$, $\beta_B$ and $s$ into $s_A$, $s_B$, and denote 
$Z= (s_{A}^\top, \beta_{B}^\top, \beta_{0}, \beta_{A}^\top, s_{B}^\top)^\top$. Then KKT conditions \eqref{KKT_vector}  can be rewritten as
\begin{equation}\label{F(Z)}
F(Z)= \left[\begin{array}{c}
\beta_{A}-\mathcal{S}(\beta_{A}+s_{A})\\
\beta_{B}-\mathcal{S}(\beta_{B}+s_{B})\\
-\frac{1}{n}1^{\top}d\\
-\frac{1}{n}X_{A}^{\top}d+\lambda\alpha s_{A}+\lambda(1-\alpha)\beta_{A}\\
-\frac{1}{n}X_{B}^{\top}d+\lambda\alpha s_{B}+\lambda(1-\alpha)\beta_{B}
\end{array}\right] = \mathbf{0}.
\end{equation}
And from \eqref{soft-thresholding decompose} we have
\begin{equation}\label{F1}
\begin{cases}
\beta_{A}-\mathcal{S}(\beta_{A}+s_{A}) = -s_{A}+\text{sgn}(\beta_{A}+s_{A}),\\
\beta_{B}-\mathcal{S}(\beta_{B}+s_{B}) = \beta_{B}
.
\end{cases}
\end{equation}

Let $\psi_\gamma$ be as in \eqref{psi}, and for brevity denote $\Psi = \Psi(\beta_0, \beta) = \frac{1}{n}\text{diag}(\psi_\gamma(~y_1 -\beta_0-x_1^\top\beta),\ldots,\psi_\gamma(y_n-\beta_0-x_n^\top\beta))$. Then the following result gives a proper Newton derivative of $F(Z)$.

\begin{theorem}\label{slant_huber}
$F(Z)$ is Newton differentiable for any $Z \in \mathbb{R}^{2p+1}$ and
\begin{eqnarray*}
H(Z)&\coloneqq&\left[\begin{array}{ccccc}
-I_{|A|} & \mathbf{0} & 0 & \mathbf{0} & \mathbf{0}\\
\mathbf{0} & I_{|B|} & 0 & \mathbf{0} & \mathbf{0}\\
\mathbf{0} & \mathbf{1}_{n}^{\top}\Psi X_{B} & \mathbf{1}_{n}^{\top}\Psi\mathbf{1}_{n} & \mathbf{1}_{n}^{\top}\Psi X_{A} & \mathbf{0}\\
\lambda\alpha I_{|A|} & X_{A}^{\top}\Psi X_{B} & X_{A}^{\top}\Psi\mathbf{1}_{n} & X_{A}^{\top}\Psi X_{A}+\lambda(1-\alpha)I_{|A|} & \mathbf{0}\\
\mathbf{0} & X_{B}^{\top}\Psi X_{B}+\lambda(1-\alpha)I_{|B|} & X_{B}^{\top}\Psi\mathbf{1}_{n} & X_{B}^{\top}\Psi X_{A} & \lambda\alpha I_{|B|}
\end{array}\right] \\
&\in&\nabla_NF(Z).
\end{eqnarray*}
Furthermore, for any $\gamma>0$ and $\alpha \in (0, 1)$, on the set $\{Z = (s, \beta_0, \beta):\;\text{ there exists } i \in\{1,\ldots,n\}\;\text{such that}\;|y_{i}-\beta_{0}-x_{i}^{\top}\beta|\leq\gamma\}$, $H(Z)$ is invertible and $H(Z)^{-1}$ is uniformly bounded in spectral norm.
\end{theorem}

From Theorems \ref{superlinear} and \ref{slant_huber}, we immediately obtain the following result.

\begin{theorem}\label{SNA_huber}
Given $\lambda, \gamma, \alpha \in (0,1)$, define $Z$ and $F(Z)$ as \eqref{F(Z)}. Suppose $\widehat{Z}$ solves $F(Z)=0$ and there exists a neighborhood $\mathcal{N}(\widehat{Z})$ such that for any $Z \in \mathcal{N}(\widehat{Z})$ there is an $i\in\{1,\ldots,n\}$ that satisfies $|y_{i}-\beta_{0}-x_{i}^{\top}\beta|\leq\gamma$, then the Newton-type iteration
\[
Z^{k+1}=Z^{k}-H(Z^{k})^{-1}F(Z^{k})
\]
converges superlinearly to $\widehat{Z}$ provided that $\|Z^0 - \widehat{Z}\|_2$ is sufficiently small.
\end{theorem}

Now we describe the algorithm in details. The $(k+1)$-th iteration can be split into two steps:
\begin{enumerate}
\item
Solve $D_k$ from $H(Z^k) D_k = -F(Z^k)$;
\item
Update $Z^{k+1} = Z^k + D^k$.
\end{enumerate}

At the first glance, step 1 seems to involve inverting a $(2p+1)\times (2p+1)$ matrix, which is intractable in high dimensional settings. However, the definitions of sets $A, B$ in \eqref{ABsets} motivate an ``active set strategy" for dimension reduction. Given the estimates from the $k$th iteration, define the active set $A_k$ and its complement $B_k$ by \eqref{AB_SNA}, $d_k=d(\beta_0^k,\beta^k)$, and $D_k = (D_{A_k}^{s\top}, D_{B_k}^{\beta\top}, D_{0}^{\beta_0}, D_{A_k}^{\beta\top}, D_{B_k}^{s\top})^\top$ corresponding to $Z_k$.

Now substituting these identities into step 1 and combining \eqref{F1} we have
\begin{eqnarray*}
D_{A_{k}}^{s} & = & -s_{A_{k}}^{k}+\text{sgn}(\beta_{A_{k}}^{k}+s_{A_{k}}^{k}),\\
D_{B_{k}}^{\beta} & = & -\beta_{B_{k}}^{k},\\
\left[\begin{array}{c}
D_{0}^{\beta_{0}}\\
D_{A_{k}}^{\beta}
\end{array}\right] & = & \left[\begin{array}{cc}
\mathbf{1}_{n}^{\top}\Psi_{k}\mathbf{1}_{n} & \mathbf{1}_{n}^{\top}\Psi_{k}X_{A_{k}}\\
X_{A_{k}}^{\top}\Psi_{k}\mathbf{1}_{n} & X_{A_{k}}^{\top}\Psi_{k}X_{A_{k}}+\lambda(1-\alpha)I_{|A_{k}|}
\end{array}\right]^{-1}\\
 &  & \left[\begin{array}{c}
\frac{1}{n}1^{\top}d_{k}+\mathbf{1}_{n}^{\top}\Psi_{k}X_{B_{k}}\beta_{B_{k}}^{k}\\
\frac{1}{n}X_{A_{k}}^{\top}d_{k}-\lambda(1-\alpha)\beta_{A_{k}}^{k}-\lambda\alpha\text{sgn}(\beta_{A_{k}}^{k}+s_{A_{k}}^{k})+X_{A_{k}}^{\top}\Psi_{k}X_{B_{k}}\beta_{B_{k}}^{k}
\end{array}\right],\\
D_{B_{k}}^{s} & = & -s_{B_{k}}^{k}+ \frac{1}{\lambda\alpha}X_{B_{k}}^{\top}
\left(\frac{1}{n}d_{k}+
\Psi_{k}X_{B_{k}}\beta_{B_{k}}^{k}-
\Psi_{k}\mathbf{1}_{n}D_{0}^{\beta_{0}}+
\Psi_{k}X_{A_{k}}D_{A_{k}}^{\beta}\right).
\end{eqnarray*}

Combining steps 1 and 2, the $(k+1)$th iteration of SNA is carried out as follows:

\begin{enumerate}[(i)]
\item Update $s_{A_k}^{k+1}$ and $\beta^{k+1}_{B_k}$:
\begin{eqnarray*}
s_{A_{k}}^{k+1} & = & \text{sgn}(\beta_{A_{k}}^{k}+s_{A_{k}}^{k}),\\
\beta_{B_{k}}^{k+1} & = & \mathbf{0}.
\end{eqnarray*}

\item Find the direction $D_0^{\beta_0}$ for the intercept $\beta_0$, and $D_{A_k}^\beta$ for the active coefficients $\beta_{A_k}$:
\begin{eqnarray*}
\left[\begin{array}{c}
D_{0}^{\beta_{0}}\\
D_{A_{k}}^{\beta}
\end{array}\right] & = & \left[\begin{array}{cc}
\mathbf{1}_{n}^{\top}\Psi_{k}\mathbf{1}_{n} & \mathbf{1}_{n}^{\top}\Psi_{k}X_{A_{k}}\\
X_{A_{k}}^{\top}\Psi_{k}\mathbf{1}_{n} & X_{A_{k}}^{\top}\Psi_{k}X_{A_{k}}+\lambda(1-\alpha)I_{|A_{k}|}
\end{array}\right]^{-1}\\
 &  & \left[\begin{array}{c}
\frac{1}{n}1^{\top}d_{k}+\mathbf{1}_{n}^{\top}\Psi_{k}X_{B_{k}}\beta_{B_{k}}^{k}\\
\frac{1}{n}X_{A_{k}}^{\top}d_{k}-\lambda(1-\alpha)\beta_{A_{k}}^{k}-\lambda\alpha s_{A_{k}}^{k+1}+X_{A_{k}}^{\top}\Psi_{k}X_{B_{k}}\beta_{B_{k}}^{k}
\end{array}\right].
\end{eqnarray*}

\item Update the intercept, the
active coefficients, and the inactive subgradients:
\begin{eqnarray*}
\beta_{0}^{k+1} & = & \beta_{0}^{k}+D_{0}^{\beta_{0}},\\
\beta_{A_{k}}^{k+1} & = & \beta_{A_{k}}^{k}+D_{A_{k}}^{\beta},\\
s_{B_{k}}^{k+1} & = & \frac{1}{\lambda\alpha}X_{B_{k}}^{\top}
\left(\frac{1}{n}d_{k}+\Psi_{k}X_{B_{k}}\beta_{B_{k}}^{k}-
\Psi_{k}\mathbf{1}_{n}D_{0}^{\beta_{0}}+
\Psi_{k}X_{A_{k}}D_{A_{k}}^{\beta}\right).
\end{eqnarray*}
\end{enumerate}

\subsection{Proofs}
Here we give proofs of Theorems \ref{convergence_lad}, \ref{convergence_ls} in the main text and Lemmas \ref{basic}, \ref{piecewise} and 
Theorem \ref{slant_huber} in the appendices. The proof of Theorem \ref{convergence_quantile} is similar to that of Theorem \ref{convergence_lad} and hence omitted.

\subsubsection*{Proof of Theorem \ref{convergence_lad}.}

\begin{proof}

Without loss of generality, assume $\theta_k$ has exactly one cluster point $\theta^\star$, i.e. $\theta_k \rightarrow \theta^\star$. Notice that
\[
|t|-\frac{\gamma}{2}\leq h_{\gamma}(t)\leq|t|,
\]
hence
\[
f_{A}(\theta;\lambda)-\frac{\gamma}{2}\leq f_{H}(\theta;\lambda,\gamma)\leq f_{A}(\theta;\lambda).
\]
Let $\widehat{\theta}_A$ be a minimizer of $f_A(\theta;\lambda)$, and $f_A^0 = \underset{\theta}{\min} f_A(\theta;\lambda) = f_A(\widehat{\theta}_A;\lambda)$, then
\[
f_{H}(\theta_{k};\lambda,\gamma_{k})\leq f_{H}(\widehat{\theta}_{A};\lambda,\gamma_{k})\leq f_{A}(\widehat{\theta}_{A};\lambda)=f_{A}^{0}.
\]
For any $\epsilon > 0$, there exists $K$ such that for $k \geq K$,
$\gamma_k < 2\epsilon$, then
\[
f_{H}(\theta_{k};\lambda,\gamma_{k})\geq f_{A}(\theta_{k};\lambda)-\epsilon\geq f_{A}^{0}-\epsilon.
\]
Hence for $k \geq K$,
\[
f_{A}^{0}-\epsilon\leq f_{A}(\theta_{k};\lambda)-\epsilon\leq f_{A}^{0}.
\]
Let $k \rightarrow \infty$, we have
$f_A^0 \leq f_A(\theta^\star) \leq f_A^0 + \epsilon$.
Since $\epsilon$ is arbitrary, $f_A(\theta^\star) = f_A^0$.
\end{proof}

\subsubsection*{Proof of Theorem \ref{convergence_ls}.}

\begin{proof}
Without loss of generality, assume $\theta_k$ has exactly one cluster point $\theta^\star$, i.e. $\theta_k \rightarrow \theta^\star$. Notice that
\[
\gamma h_{\gamma}(t)\leq \frac{1}{2}t^2,
\]
which implies
\[
\gamma f_{H}(\theta;\lambda/\gamma,\gamma)\leq f_{S}(\theta;\lambda).
\]
Let $\widehat{\theta}_S$ be a minimizer of $f_S(\theta;\lambda)$, and $f_S^0 = \underset{\theta}{\min} f_S(\theta;\lambda) = f_S(\widehat{\theta}_S;\lambda)$, then
\[
\gamma_k f_{H}(\theta_{k};\lambda/\gamma_k,\gamma_{k})\leq \gamma_k f_{H}(\widehat{\theta}_{S};\lambda/\gamma_k,\gamma_{k})\leq f_{S}(\widehat{\theta}_{S};\lambda)=f_{S}^{0}.
\]
Since $\theta_k = (\beta_0^k,\beta^k)$ is convergent, $r^k = y - \beta_0^k\texttt{1} -X\beta^k$ is convergent too. Then there exists $M>0$ such that $\|r^k\|_{\infty}\leq M$. There exists $K$ such that for $k \geq K$,
$\gamma_k > M$, then $h_\gamma(r_ik) = \frac{1}{2}r_ik^2$, and
\[
\gamma_kf_{H}(\theta_{k};\lambda/\gamma_k,\gamma_{k})=f_S(\theta_k; \lambda).
\]
Hence for $k \geq K$,
\[
f_S(\theta_k; \lambda)\leq f_{S}^{0}.
\]
Let $k \rightarrow \infty$, we have
$f_S(\theta^\star;\lambda) \leq f_S^0$. Since $f_S(\theta^\star;\lambda) \geq \underset{\theta}{\min} f_S(\theta;\lambda) = f_S^0$,
 $f_S(\theta^\star) = f_S^0$.
\end{proof}

\subsubsection*{Proof of Lemma \ref{basic}.}

\begin{proof}
\begin{enumerate}[(i)]
\item
By assumption, the Jacobian $J_F$ is continuous at $z$.
Since
\begin{eqnarray*}
\lefteqn{\frac{\|F(z+h)-F(z)-J_{F}(z+h)h\|_{2}}{\|h\|_{2}}}\\
 & \leq & \frac{\|F(z+h)-F(z)-J_{F}(z)h\|_{2}+\|(J_{F}(z)-J_{F}(z+h))h\|_{2}}{\|h\|_{2}}\\
 & \leq & \frac{\|F(z+h)-F(z)-J_{F}(z)h\|_{2}}{\|h\|_{2}}+\|J_{F}(z)-J_{F}(z+h)\|\\
 & \rightarrow & 0
\end{eqnarray*}
as $h \rightarrow \textbf{0}$, by definition $J_F \in \nabla_N F(z)$.

\item
\[
\|AF(z+h)-AF(z)-AH(z+h)h\|_{2}\leq\|A\|\|F(z+h)-F(z)-H(z+h)h\|_{2}=o(\|h\|_{2}),
\]
hence $AH \in \nabla_N AF(z)$.

\item
\[
\begin{array}{ll}
 & \|(F(z+h)+G(z+h))-(F(z)+G(z))-(H_{F}(z+h)+H_{G}(z+h))h\|_{2}\\
\leq & \|F(z+h)-F(z)-H_{F}(z+h)h\|_{2}+\|G(z+h)-G(z)-H_{G}(z+h)h\|_{2}\\
= & o(\|h\|_{2}),
\end{array}
\]
hence $H_F+H_G\in \nabla_N (F+G)(z)$.

\item
It can be seen by observing that
\[
\|F(z+h)-F(z)-H(z+h)h\|_{2}^{2}=\sum_{i=1}^{l}(F_{i}(z+h)-F_{i}(z)-H_{i}(z+h)h)^{2}.
\]
\end{enumerate}
\end{proof}

\subsubsection*{Proof of Lemma \ref{piecewise}.}

\begin{proof}
If $f$ is differentiable at $z$ with derivative $f^\prime$ defined in its neighborhood, by smoothness assumption and Lemma \ref{basic}(i), $f^\prime \in \nabla_N f(z)$.

If $f$ is not differentiable at $z$, by assumption there exists $s>0$ such that $f$ is smooth on both $(z-s, z)$ and $(z, z+s)$ implying that
$f^{\prime}(z-)=\lim_{h\rightarrow0^{-}}\frac{f(z+h)-f(z)}{h}$ and $f^{\prime}(z+)=\lim_{h\rightarrow0^{+}}\frac{f(z+h)-f(z)}{h}$ exist and
\begin{eqnarray*}
f^{\prime}(z+h)\rightarrow f^{\prime}(z-) &  & \text{as }h\rightarrow0^{-},\\
f^{\prime}(z+h)\rightarrow f^{\prime}(z+) &  & \text{as }h\rightarrow0^{+}.
\end{eqnarray*}
Hence for any $\varepsilon>0$, there exists a sufficiently small $\delta>0$ such that
\begin{eqnarray*}
\forall\:x\in(z-\delta,z), & \frac{|f(x)-f(z)-f^{\prime}(z-)(x-z)|}{|x-z|}<\varepsilon/2, & |f^{\prime}(x)-f^{\prime}(z-)|<\varepsilon/2;\\
\forall\:x\in(z,z+\delta), & \frac{|f(x)-f(z)-f^{\prime}(z+)(x-z)|}{|x-z|}<\varepsilon/2, & |f^{\prime}(x)-f^{\prime}(z+)|<\varepsilon/2.
\end{eqnarray*}
Thus for $x\in (z-\delta, z)$,
\[
\frac{|f(x)-f(z)-f^{\prime}(x)(x-z)|}{|x-z|}\leq\frac{|f(x)-f(z)-f^{\prime}(z-)(x-z)|}{|x-z|}+|f^{\prime}(z-)-f^{\prime}(x)|<\varepsilon,
\]
and similarly for $x \in (z, z+\delta)$.
Define $H(z)$ as in the lemma, then the above implies

\[
\forall\varepsilon>0,\exists\delta>0\text{ s.t. }\forall|x-\delta|<z,\;\frac{|f(x)-f(z)-H(x)(x-z)|}{|x-z|}<\epsilon.
\]
In other word, $f$ is Newton differentiable at $z$ with $H \in \nabla_N f(z)$.

\end{proof}

In order to prove Theorem \ref{slant_huber}, we need the following lemma.

\begin{lemma}\label{h3_bound}
Given $\alpha \in (0, 1)$ and $\beta_0, \beta$ satisfy $|y_i-\beta_0-x_i^\top \beta|\leq \gamma$ for some $i$, then $H_3$ in \eqref{h3} is invertible with its inverse uniformly bounded in spectral norm, i.e.
\[
\|H_{3}^{-1}\|\leq\frac{1}{\lambda\alpha}+\left[\frac{1}{\lambda(1-\alpha)}+\frac{\lambda_{\max}(X^{\top}X)^{2}+n\gamma\lambda(1-\alpha)}{\lambda(1-\alpha)}\left(1+\frac{\|X\|}{\sqrt{n}\gamma\lambda(1-\alpha)}\right)^{2}\right]\left(1+\frac{2\|X\|}{\sqrt{n}\gamma\lambda\alpha}\right).
\]
\end{lemma}

\begin{proof}
Denote $ J = n\gamma\Psi$, then $J$ is diagonal and idempotent. We have
\[
\mathbf{1}_{n}^{\top}\Psi\mathbf{1}_{n}=\frac{1}{n\gamma}\mathbf{1}_{n}^{\top}J\mathbf{1}_{n}=\frac{1}{n\gamma}(J\mathbf{1}_{n})^{\top}(J\mathbf{1}_{n}),
\]
and
\begin{eqnarray*}
\lefteqn{\mathbf{1}_{n}^{\top}\Psi X_{A}\left(X_{A}^{\top}\Psi X_{A}+\lambda(1-\alpha)I_{|A|}\right)^{-1}X_{A}^{\top}\Psi\mathbf{1}_{n}}\\
&=&\frac{1}{n\gamma}(J\mathbf{1}_{n})^{\top}(JX_{A})
\left((JX_{A})^{\top}(JX_{A})+n\gamma\lambda(1-\alpha)I_{|A|}\right)^{-1}
(JX_{A})^{\top}(J\mathbf{1}_{n}).
\end{eqnarray*}

Denote $a = J\mathbf{1}_n$, $Z=JX_A$, $t=n\gamma\lambda(1-\alpha)$, and $m=|A|$. Then the LHS becomes
\[
\frac{1}{n\gamma}\bigg(a^{\top}a-a^{\top}Z(Z^{\top}Z+tI_{m})^{-1}Z^{\top}a\bigg).
\]

Since $|y_i-\beta_0-x_i^\top \beta|\leq \gamma$ for some $i$, we have $\psi_i = \frac{1}{n\gamma}>0$, implying that $J_{ii}=1$ and $a^\top a \geq J_{ii}^2 = 1$. Thus we are guaranteed that $a=J\mathbf{1}_n$ is not a zero vector.

Now apply SVD to $Z$ such that
$Z=UDV^\top$, where $U_{n\times n}$ and $V_{m\times m}$ are both orthogonal matrices, and $D_{n\times m}$ is a rectangular diagonal matrix with non-negative diagonal elements $d_1, \ldots, d_{m\wedge n}$.
Hence
\begin{eqnarray*}
Z(Z^{\top}Z+tI_{m})^{-1}Z^{\top} & = & UDV^{\top}(VD^{\top}U^{\top}UDV^{\top}+tI_{m})^{-1}VD^{\top}U^{\top}\\
 & = & UDV^{\top}\left(V(D^{\top}D+tI_{m})V^{\top}\right)^{-1}VD^{\top}U^{\top}\\
 & = & UDV^{\top}V(D^{\top}D+tI_{m})^{-1}V^{\top}VD^{\top}U^{\top}\\
 & = & UD(D^{\top}D+tI_{m})^{-1}D^{\top}U^{\top}.
\end{eqnarray*}
When $n>m$,
\[
D(D^{\top}D+tI_{m})^{-1}D^{\top}=\text{diag}(\frac{d_{1}^{2}}{d_{1}^{2}+t},\ldots,\frac{d_{m}^{2}}{d_{m}^{2}+t},0,\ldots,0),
\]
and when $n\leq m$,
\[
D(D^{\top}D+tI_{m})^{-1}D^{\top}=\text{diag}(\frac{d_{1}^{2}}{d_{1}^{2}+t},\ldots,\frac{d_{n}^{2}}{d_{n}^{2}+t}).
\]
In either case $D(D^{\top}D+tI_{m})^{-1}D^{\top}$ is p.s.d. with $\lambda_{\max}(D(D^{\top}D+tI_{m})^{-1}D^{\top})< 1$.

Next we will derive the upper bound of eigenvalues of the above matrix.
First, for any eigenvalue $d$ and corresponding nonzero eigenvector $u$ of $Z^\top Z = X_AJX_A$, we have
\[
du^{\top}u=u^{\top}X_{A}^{\top}JX_{A}u=\left[\begin{array}{c}
u\\
0
\end{array}\right]^{\top}X^{\top}JX\left[\begin{array}{c}
u\\
0
\end{array}\right]\leq\lambda_{\max}(X^{\top}JX)u^{\top}u,
\]
hence $d\leq \lambda_{\max}(X^{\top}JX)$. Then again, for any eigenvalue $c$ and corresponding nonzero eigenvector $v$ of $X^{\top}JX$, we have
\[
cv^{\top}v=v^{\top}X^{\top}JXv=\sum_i J_{ii}v^{\top}x_{i}x_{i}^{\top}v\leq\sum_i v^{\top}x_{i}x_{i}^{\top}v=v^{\top}X^{\top}Xv\leq\lambda_{\max}(X^{\top}X)v^{\top}v,
\]

implying that $c\leq \lambda_{\max}(X^{\top}X)$.

Therefore, we have $d\leq \lambda_{\max}(X^{\top}JX) \leq \lambda_{\max}(X^{\top}X)$. Then since the eigenvalues of $Z^\top Z$ are the diagonal elements of $D$, the eigenvalues of $D(D^{\top}D+tI_{m})^{-1}D^{\top}$ are bounded by $\frac{\lambda_{\max}(X^{\top}X)^{2}}{\lambda_{\max}(X^{\top}X)^{2}+t}$.

Then recall $t = n\gamma \lambda (1-\alpha)$ and $a^\top a \geq 1$, we have
\begin{eqnarray*}
\lefteqn{\mathbf{1}_{n}^{\top}\Psi\mathbf{1}_{n}-\mathbf{1}_{n}^{\top}\Psi X_{A}\left(X_{A}^{\top}\Psi X_{A}+\lambda(1-\alpha)I_{|A|}\right)^{-1}X_{A}^{\top}\Psi\mathbf{1}_{n} }\\
& = & \frac{1}{n\gamma}\left(a^{\top}a-(U^{\top}a)^{\top}D(D^{\top}D+tI_{m})^{-1}D^{\top}(U^{\top}a)\right)\\
 & \geq & \frac{1}{n\gamma}(a^{\top}a-\frac{\lambda_{\max}(X^{\top}X)^{2}}{\lambda_{\max}(X^{\top}X)^{2}+t}(U^{\top}a)^{\top}U^{\top}a)\\
 & = & \frac{1}{n\gamma}\times\frac{n\gamma\lambda(1-\alpha)}{\lambda_{\max}(X^{\top}X)^{2}+n\gamma\lambda(1-\alpha)}a^{\top}a\\
 & \geq & \frac{\lambda(1-\alpha)}{\lambda_{\max}(X^{\top}X)^{2}+n\gamma\lambda(1-\alpha)}\\
 & > & 0.
\end{eqnarray*}

Let
\[
H_{31}=\left[\begin{array}{cc}
\mathbf{1}_{n}^{\top}\Psi\mathbf{1}_{n} & \mathbf{1}_{n}^{\top}\Psi X_{A}\\
X_{A}^{\top}\Psi\mathbf{1}_{n} & X_{A}^{\top}\Psi X_{A}+\lambda(1-\alpha)I_{|A|}
\end{array}\right],\; H_{32}=\left[\begin{array}{cc}
X_{B}^{\top}\Psi\mathbf{1}_{n} & X_{B}^{\top}\Psi X_{A}\end{array}\right],\; H_{33}=\lambda\alpha I_{|B|}.
\]
Observe that $H_{33}^{-1}=\frac{1}{\lambda\alpha}I_{|B|}$. Then if $H_{31}$ is invertible, we have
\[
H_{3}^{-1}=\left[\begin{array}{cc}
H_{31}^{-1} & \mathbf{0}\\
-\frac{1}{\lambda\alpha}H_{32}H_{31}^{-1} & \frac{1}{\lambda\alpha}I_{|B|}
\end{array}\right].
\]
Hence to show $H_3$ is invertible, it suffices to show $H_{31}$ is invertible.
Let
\[
M=X_{A}^{\top}\Psi X_{A}+\lambda(1-\alpha)I_{|A|},\; b=X_{A}^{\top}\Psi\mathbf{1}_{n},
\]
and
\[
\kappa=\mathbf{1}_{n}^{\top}\Psi\mathbf{1}_{n}-\mathbf{1}_{n}^{\top}\Psi X_{A}\left(X_{A}^{\top}\Psi X_{A}+\lambda(1-\alpha)I_{|A|}\right)^{-1}X_{A}^{\top}\Psi\mathbf{1}_{n}.
\]

Since $\kappa > 0$, we have
\[
H_{31}^{-1}=\left[\begin{array}{cc}
\frac{1}{\kappa} & -\frac{1}{\kappa}b^{\top}M^{-1}\\
-\frac{1}{\kappa}M^{-1}b & M^{-1}+\frac{1}{\kappa}M^{-1}bb^{\top}M^{-1}
\end{array}\right],
\]
and it follows that $H_3$ is invertible.

It can be easily shown that $\|b\|=\|b^{\top}\|\leq\frac{1}{\sqrt{n}\gamma}\|X\|,\;\|M^{-1}\|\leq\frac{1}{\lambda(1-\alpha)}$. Combine this with $\frac{1}{\kappa}\leq\frac{\lambda_{\max}(X^{\top}X)^{2}+n\gamma\lambda(1-\alpha)}{\lambda(1-\alpha)}$, then similar to \eqref{bound}, we have
\[
\|H_{31}^{-1}\|\leq\frac{1}{\lambda(1-\alpha)}+\frac{\lambda_{\max}(X^{\top}X)^{2}+n\gamma\lambda(1-\alpha)}{\lambda(1-\alpha)}\left(1+\frac{\|X\|}{\sqrt{n}\gamma\lambda(1-\alpha)}\right)^{2}
\]

and then
\[
\|H_{3}^{-1}\|\leq\frac{1}{\lambda\alpha}+\left[\frac{1}{\lambda(1-\alpha)}+\frac{\lambda_{\max}(X^{\top}X)^{2}+n\gamma\lambda(1-\alpha)}{\lambda(1-\alpha)}\left(1+\frac{\|X\|}{\sqrt{n}\gamma\lambda(1-\alpha)}\right)^{2}\right]\left(1+\frac{2\|X\|}{\sqrt{n}\gamma\lambda\alpha}\right).
\]
\end{proof}

\subsubsection*{Proof of Theorem \ref{slant_huber}.}

\begin{proof}
Notice $\mathcal{S}$ is piecewise-smooth, then by Lemma \ref{chain}, \ref{piecewise} and Lemma \ref{basic} (iv) $F_1 (Z)$ is Newton differentiable, and with \eqref{F1}
\[
\left[\begin{array}{ccccc}
-I_{|A|} & \mathbf{0} & 0 & \mathbf{0} & \mathbf{0}\\
\mathbf{0} & I_{|B|} & 0 & \mathbf{0} & \mathbf{0}
\end{array}\right]\in\nabla_NF_{1}(Z).
\]

Similarly, the Huber loss is also piecewise-smooth, and by Lemma \ref{chain}, \ref{piecewise} and Lemma \ref{basic} (ii)-(iv), we have $F_2 (Z)$ and $F_3(Z)$ are Newton differentiable and

\[
\left[\begin{array}{ccccc}
\mathbf{0} & \mathbf{1}_{n}^{\top}\Psi X_{B} & \mathbf{1}_{n}^{\top}\Psi\mathbf{1}_{n} & \mathbf{1}_{n}^{\top}\Psi X_{A} & \mathbf{0}\end{array}\right]\in\nabla_NF_{2}(Z),
\]

\[
\left[\begin{array}{ccccc}
\mathbf{0} & \mathbf{1}_{n}^{\top}\Psi X_{B} & \mathbf{1}_{n}^{\top}\Psi\mathbf{1}_{n} & \mathbf{1}_{n}^{\top}\Psi X_{A} & \mathbf{0}\\
\lambda\alpha I_{|A|} & X_{A}^{\top}\Psi X_{B} & X_{A}^{\top}\Psi\mathbf{1}_{n} & X_{A}^{\top}\Psi X_{A}+\lambda(1-\alpha)I_{|A|} & \mathbf{0}\\
\mathbf{0} & X_{B}^{\top}\Psi X_{B}+\lambda(1-\alpha)I_{|B|} & X_{B}^{\top}\Psi\mathbf{1}_{n} & X_{B}^{\top}\Psi X_{A} & \lambda\alpha I_{|B|}
\end{array}\right]\in\nabla_NF_{3}(Z).
\]

Again, by Lemma \ref{basic} (iv), $F(Z)$ is Newton differentiable and
\[
H=\left[\begin{array}{ccccc}
-I_{|A|} & \mathbf{0} & 0 & \mathbf{0} & \mathbf{0}\\
\mathbf{0} & I_{|B|} & 0 & \mathbf{0} & \mathbf{0}\\
\mathbf{0} & \mathbf{1}_{n}^{\top}\Psi X_{B} & \mathbf{1}_{n}^{\top}\Psi\mathbf{1}_{n} & \mathbf{1}_{n}^{\top}\Psi X_{A} & \mathbf{0}\\
\lambda\alpha I_{|A|} & X_{A}^{\top}\Psi X_{B} & X_{A}^{\top}\Psi\mathbf{1}_{n} & X_{A}^{\top}\Psi X_{A}+\lambda(1-\alpha)I_{|A|} & \mathbf{0}\\
\mathbf{0} & X_{B}^{\top}\Psi X_{B}+\lambda(1-\alpha)I_{|B|} & X_{B}^{\top}\Psi\mathbf{1}_{n} & X_{B}^{\top}\Psi X_{A} & \lambda\alpha I_{|B|}
\end{array}\right]\in\nabla_NF(Z).
\]

Now let
\[
H_{1}=\left[\begin{array}{cc}
-I_{|A|} & \mathbf{0}\\
\mathbf{0} & I_{|B|}
\end{array}\right],\; H_{2}=\left[\begin{array}{cc}
\mathbf{0} & \mathbf{1}_{n}^{\top}\Psi X_{B}\\
\lambda\alpha I_{|A|} & X_{A}^{\top}\Psi X_{B}\mathbf{0}X_{B}^{\top}\Psi X_{B}+\lambda(1-\alpha)I_{|B|}
\end{array}\right],
\]
\begin{equation}\label{h3}
H_{3}=\left[\begin{array}{ccc}
\mathbf{1}_{n}^{\top}\Psi\mathbf{1}_{n} & \mathbf{1}_{n}^{\top}\Psi X_{A} & \mathbf{0}\\
X_{A}^{\top}\Psi\mathbf{1}_{n} & X_{A}^{\top}\Psi X_{A}+\lambda(1-\alpha)I_{|A|} & \mathbf{0}\\
X_{B}^{\top}\Psi\mathbf{1}_{n} & X_{B}^{\top}\Psi X_{A} & \lambda\alpha I_{|B|}
\end{array}\right].
\end{equation}

Then it is clear that $H_1$ is invertible. Now if $H_3$ is also invertible, which we show in Lemma \ref{h3_bound} under a mild condition, then via some algebra we have
\begin{equation}
H^{-1}=\left[\begin{array}{cc}
H_{1}^{-1} & \mathbf{0}\\
-H_{3}^{-1}H_{2}H_{1}^{-1} & H_{3}^{-1}
\end{array}\right].
\end{equation}

Let $g = (g_1^\top, g_2^\top)^\top\in \mathbb{R}^{p}\times \mathbb{R}^{p+1}$, then
\begin{equation}\label{bound}
\begin{aligned}\|H^{-1}g\|_{2}^{2} & =\|H_{1}^{-1}g_{1}\|_{2}^{2}+\|-H_{3}^{-1}H_{2}H_{1}^{-1}g_{1}+H_{3}^{-1}g_{2}\|_{2}^{2}\\
 & \leq\|H_{1}^{-1}\|^{2}\|g_{1}\|_{2}^{2}+(\|H_{3}^{-1}\|\|H_{2}\|\|H_{1}^{-1}\|\|g_{1}\|_{2}+\|H_{3}^{-1}\|\|g_{2}\|_{2})^{2}\\
 & \leq(\|H_{1}^{-1}\|\|g_{1}\|_{2}+\|H_{3}^{-1}\|\|H_{2}\|\|H_{1}^{-1}\|\|g_{1}\|_{2}+\|H_{3}^{-1}\|\|g_{2}\|_{2})^{2}\\
 & \leq(\|H_{1}^{-1}\|+\|H_{3}^{-1}\|+\|H_{3}^{-1}\|\|H_{2}\|\|H_{1}^{-1}\|)^{2}\|g\|_{2}^{2}
\end{aligned}
\end{equation}
which implies
\begin{equation}\label{h_bound_result}
\|H^{-1}\|\leq\|H_{1}^{-1}\|+\|H_{3}^{-1}\|+\|H_{3}^{-1}\|\|H_{2}\|\|H_{1}^{-1}\|.
\end{equation}

Notice $\|X_{A}\|\lor\|X_{B}\|\leq\|X\|$. Take $X_A$, without loss of generality shuffle columns of $X$ such that
$X=\left(\begin{array}{cc} X_{A} & X_{B}\end{array}\right)$, then for any $g\in \mathbb{R}^{|A|}$ such that $\|g\|_2 = 1$, we have
\[
\|X_{A}g\|_{2}=\|X\left(\begin{array}{c}
g\\
\mathbf{0}
\end{array}\right)\|_{2}\leq\sup\left\{ \|Xv\|_{2}:\;\|v\|_{2}=1\right\} =\|X\|,
\]
implying that $\|X_{A}\|=\sup\left\{ \|X_{A}g\|_{2}:\;\|g\|_{2}=1\right\} \leq\|X\|$. Similarly for $X_B$.

Then a similar argument as in \eqref{bound} shows that
\begin{equation}\label{h2_bound}
\|H_{2}\|\leq1+\alpha+2\|X\|^{2}.
\end{equation}

Combining \eqref{h_bound_result}, \eqref{h2_bound} with results of Lemma \ref{h3_bound} under its condition, and observing that $\|H_{1}^{-1}\|=1$, we obtain the uniform boundedness of $H$ in spectral norm, i.e.,
\begin{eqnarray*}
\|H^{-1}\| &\leq & 1+
\left[\frac{1}{\lambda\alpha}+\left(\frac{1}{\lambda(1-\alpha)}+
\frac{\lambda_{\max}(X^{\top}X)^{2}+n\gamma\lambda(1-\alpha)}
{\lambda(1-\alpha)}\left(1+\frac{\|X\|}
{\sqrt{n}\gamma\lambda(1-\alpha)}\right)^{2}\right) \right.\\
& & \times \left.
\left(1+\frac{2\|X\|}
{\sqrt{n}\gamma\lambda\alpha}\right)\right](2+\alpha+2\|X\|^{2}).
\end{eqnarray*}
\end{proof}

\renewcommand{\refname}{REFERENCES}
\spacingset{0}
\small
\bibliography{huber}
\end{document}